\documentclass[12pt]{article}

\usepackage{amssymb,amsfonts, amsmath, bm, eurosym,geometry,ulem,graphicx,caption,mathtools,natbib, xcolor,setspace,comment,array, parskip, bm, csquotes, multirow, bbm, titlesec, footmisc, color, amsthm, subcaption, enumitem, mlmodern}

\usepackage[colorlinks=true]{hyperref}
\usepackage{cleveref}
\usepackage{footnotebackref}
\usepackage[title,titletoc]{appendix}
\usepackage[flushleft]{threeparttable}
\usepackage[capposition=top]{floatrow}

\renewcommand{\arraystretch}{1.22}

\let\oldcite\cite
\renewcommand{\cite}[1]{\textcolor{blue}{\oldcite{#1}}}

\titleformat{\section}
  {\centering \normalsize \scshape}{\thesection.}{0.5em}{}

\titleformat{name=\section,numberless}
  {\normalfont \scshape}{}{0em}{}

\titleformat{\subsection}
  {\normalsize \scshape}{\thesubsection}{0.5em}{}

\titleformat{name=\subsection,numberless}
  {\normalfont \scshape}{}{0em}{}

\makeatletter
\renewcommand\hyper@natlinkbreak[2]{#1}
\makeatother

\graphicspath{ {./images/} }

\newtheorem{theorem}{Theorem}

\newtheorem{lemma}[theorem]{Lemma}

\makeatletter
\renewenvironment{proof}[1][\proofname]{\par
  \pushQED{\qed}%
  \normalfont \topsep6\p@\@plus6\p@\relax
  \trivlist
  \item[\hskip\labelsep
        \bfseries
    #1\@addpunct{.}]\ignorespaces
}{%
  \popQED\endtrivlist\@endpefalse
}
\makeatother

\newcolumntype{L}[1]{>{\raggedright\let\newline\\arraybackslash\hspace{0pt}}m{#1}}
\newcolumntype{C}[1]{>{\centering\let\newline\\arraybackslash\hspace{0pt}}m{#1}}
\newcolumntype{R}[1]{>{\raggedleft\let\newline\\arraybackslash\hspace{0pt}}m{#1}}

  {\list{}{\leftmargin=0.3in\rightmargin=0in}\item[]}%
  {\endlist}

 \hypersetup{
    pdftoolbar=true,        % show Acrobat’s toolbar?
    pdfmenubar=true,        % show Acrobat’s menu?
    pdffitwindow=false,     % window fit to page when opened
    pdfstartview={FitH},    % fits the width of the page to the window
    pdftitle={title},    % title
    pdfauthor={Salvatore Mazzarino},     % author
    pdfsubject={Subject},   % subject of the document
    pdfcreator={Salvatore Mazzarino},   % creator of the document
    pdfproducer={Salvatore Mazzarino}, % producer of the document
    pdfkeywords={Green Networking} {Mobile Cloud} {Network Coding} {Energy}, % list of keywords
    pdfnewwindow=true,      % links in new window
    colorlinks=true,       % false: boxed links; true: colored links
    linkcolor=blue,          % color of internal links (change box color with linkbordercolor)
    citecolor=blue,        % color of links to bibliography
    filecolor=blue,      % color of file links
    urlcolor=blue           % color of external links
}

\makeatletter
\renewcommand\hyper@natlinkbreak[2]{#1}
\makeatother

\begin{document}

\begin{titlepage}
\title{\vspace{-0.cm}\Large How manipulable are prediction markets?\vspace{-0.55cm}\thanks{For comments on our experimental design, we are grateful to Jesper Akesson, Adam Brzezinski, Robin Hanson, Junnan He and Koleman Strumpf. For additional comments and suggestions, we would like to thank Sam Altmann, Johannes Abeler, Miguel Ballester, Gabriele Camera, Jeanne Commault, Antoine Ferey, Laure Goursat, Robert Hahn, Emeric Henry, Marco Mantovani, Théo Marquis, Ryan Oprea, Franz Ostrizek, Lionel Page, Eduardo Perez-Richet, Stefan Pollinger, Daniela Puzzello, Catriona Scott, Alessandro Stringhi, Jasmine Theilgaard, Séverine Toussaert and Eugenio Verrina. We are also grateful to audiences at Sciences Po, the University of Oxford, the Lisbon School of Economics and Management, City St George's, and the Panthéon-Sorbonne Prediction Markets Workshop. This project received IRB approval from the Paris School of Economics (2023-042).
\newline [\href{https://www.socialscienceregistry.org/trials/12543}{Main pre-registration}; \href{https://www.socialscienceregistry.org/trials/14767}{Follow-up}] [\href{https://github.com/Itzhak95/prediction_markets_model}{Simulations}] \href{https://drive.google.com/drive/folders/1lAD3E9djYVvm6DzXXx9hPN0Mq1k4OKGK?usp=sharing}{[Replication package]}}}
\author{Itzhak Rasooly\thanks{Sciences Po \& Paris School of Economics.}
\and
Roberto Rozzi\thanks{University of Siena.}
}
\date{\today}
\maketitle
\vspace{-2em}
\begin{abstract}
\noindent \textsc{Abstract.} In this paper, we conduct a large-scale field experiment to investigate the manipulability of prediction markets. The main experiment involves randomly shocking prices across 817 separate markets; we then collect hourly price data to examine whether the effects of these shocks persist over time. We find that prediction markets can be manipulated: the effects of our trades are visible even 60 days after they have occurred. However, as predicted by our model, the effects of the manipulations somewhat fade over time. Markets with more traders, greater trading volume, and an external source of probability estimates are harder to manipulate.

\noindent \\
\vspace{-0.75cm}
\vspace{0in}

\bigskip
\end{abstract}
\setcounter{page}{0}
\thispagestyle{empty}
\end{titlepage}
\pagebreak
\newpage

\section{Introduction}\label{sec_intro}

Prediction markets --- that is, markets in which traders can bet on the outcomes of various events --- have proven to be powerful information aggregators. Existing research shows that the prices that arise in prediction markets can predict outcomes as well if not better than alternative forecasting methods such as expert polls \citep{figlewski1979subjective, roll1984orange, pennock2001real, wolfers2002three, berg2008results, dreber2015using}. In addition, while alternative methods such as polling can be costly, prediction markets are largely self-financing. In part for this reason, prediction markets currently constitute the \textit{only} source of probability estimates on many important questions. For example, although one can find prediction markets on the outcomes of various geopolitical conflicts, it is difficult to obtain public probability estimates on the outcomes of these conflicts through other means.

Perhaps due to these encouraging results, prediction markets are currently undergoing something of a renaissance. Although the prediction markets previously studied in the literature (e.g. the Iowa Electronic Markets) were often small and developed for academic purposes, recent years have seen the founding of several much larger prediction market platforms. Polymarket, founded in 2020, often handles large volumes of trade: for example, $\sim$\$7 billion was traded on Polymarket in 2024. Kalshi, founded in 2021, also routinely handles large volumes of trade, peaking at $\sim$\$1 billion in 2024. Manifold Markets, also founded in 2021, has become the largest prediction market website in the world as measured by the number of markets hosted on the platform.

Despite this promise, prediction markets are hampered by long-standing concerns about manipulability. Historically, it has not been uncommon for politicians to bet on themselves in the hope of increasing the market's assessment of their electoral winning chances \citep{rhode2004historical, rhode2006manipulating}. While this is one motivation for individuals to attempt to manipulate prediction market prices, it is not the only one: for example, concerns that terrorists would make trades to `distract' onlookers from their true targets were cited during the cancellation of the planned Policy Analysis Market in 2003.\footnote{In the words of one Nobel Laureate, `[trading] could be
subject to manipulation, particularly if the market has few participants --- providing a false
sense of security or an equally false sense of alarm' \citep{stiglitz}. While this concern may be valid in general, it is worth noting that the Policy Analysis Market would not have allowed traders to bet on the locations or timings of terrorist events \citep{hanson2004foul, hanson2009manipulator}.} Concerns about the manipulability of prediction markets remain prominent in more recent media coverage, not least during the run-up to the 2024 US Presidental Election: see, for example, \cite{vox}, \cite{time}, \cite{ft_election}, \cite{wsj}, \cite{bloomberg} and \cite{nyt2}.

Although it is clear that many would \textit{like} to manipulate prediction markets, it is less clear that such attempts would be successful. Indeed, according to much of the existing literature (see, e.g., \citealp{wolfers2004prediction, berg2006iowa, sunstein, crane2022comment}), any manipulation attempt will be extremely short-lived since it will be rapidly `undone' by the behavioural responses of subsequent traders. Understanding this issue also provides an indirect test of the efficiency of prediction markets \citep{fama1970efficient}. If market prices only reflect the `fundamentals', then the effects of random trades should be transient. However, if markets are inefficient, the effect of random trades could persist.

In this paper, we study these questions using both theoretical and experimental methods. On the theoretical side, we build the first model that traces out the path of prices in a prediction market following a manipulation attempt. On the experimental side, we conduct the first large-scale field experiment on the manipulability of prediction markets.

We begin by presenting our model. In the model, agents are risk-averse and optimally choose how to bet based on their beliefs about the likelihood that the event will take place. Prices are determined using the algorithm that underpins the markets that we study experimentally. Under this algorithm, prices are non-linear in purchase quantities; and our traders understand this when they decide how much to bet. In contrast to existing literature, our model does not assume competitive equilibrium pricing. This allows us to avoid logical inconsistencies\footnote{Specifically, competitive equilibrium assumes a continuum of traders, none of whom can move the market price \citep{aumann1964markets}; in contrast, the whole point of our model is to study the adjustment of prices after the price is pushed in a particular direction by a trade.} and to study the full path of prices following a manipulation attempt (including `disequilibrium' prices).

The model makes three important predictions about the impact of a manipulative trade. First, the model predicts that manipulation can systematically affect the market price --- both in the short run (due to the need for behavioural adjustment), but also in the long run (e.g., due to `learning effects'). Second, however, the model also predicts that the effect of manipulation should be partially `undone' by future trades: for example, if a manipulator increases the market price, then traders should believe that the price is `too high' and take actions to correct this. The speed of price reversion is predicted to decline as the price moves back towards its original value. Third, the model suggests that the degree of reversion should systematically vary by market type. For example, markets with more traders and less scope for `learning' from the market price should be harder to manipulate.

To test these predictions, we conduct the first-ever large-scale field experiment on the manipulability of prediction markets. Our main experiment takes place on the Manifold platform, which had around 10,000 active users at the time when the experiment took place. Crucially, Manifold hosted (and continues to host) a very large number of different prediction markets, a feature which makes our large-scale field experiment possible. Unusually for a prediction market platform, Manifold motivates users using a mix of financial incentives (winnings can be donated to charity) along with social and self-image incentives enhanced by gamified statistics and leaderboards. Despite these features, we show that Manifold markets are remarkably well-calibrated and exhibit levels of predictive accuracy that are comparable to those of more traditional prediction market platforms.

The basic idea of our experiment is straightforward. In every market, we either place a \textit{yes bet} (chosen to instantaneously increase the price by 5 percentage points), a \textit{no bet} (chosen to instantaneously decrease the price by 5 percentage points), or do nothing (the `control'). Notice that, unusually for a field experiment, our experiment randomises at the market level: in total, we have $n = 817$ different prediction markets in our sample. To investigate whether the effects of our bets persist, we then collect hourly price data over a 30 day period (leading to $\sim$600,000 price observations in total); we also measure a snapshot of prices after 60 days. To understand important heterogeneities, we collect rich data on each market's type including the total volume and number of traders at the time when our trade is executed.

Our main experiment yields three main sets of results. First, we show that the effects of our manipulations are visible in the data even 60 days after they have occurred: market prices in the `yes' group are higher than they would otherwise have been, and market prices in the `no' group are correspondingly lower. Second, however, as predicted by our theoretical model, the behavioural responses of subsequent traders generate some reversion of prices towards their original values. As in the model, reversion is relatively quick in the week following our bets but slows as time progresses. Third, the degree of reversion varies across markets in just the way that is suggested by our model. In particular, markets which are duplicated on a different prediction platform (which provides an `external' source of probability estimates), markets with more traders, and markets with more `activity' (measured, e.g., by the number of comments left by traders) are more difficult to manipulate.

To check the robustness of our results, we conduct an analogous follow-up experiment using a new type of market that Manifold introduced. In contrast to the markets studied in our main experiment, these markets run on a currency that is redeemable at a 1 to 1 rate with US dollars. The results of our follow-up experiment are broadly similar to those of the main experiment. We once again observe a substantial degree of reversion, i.e. that the reactions of other traders push prices back towards their original values. However, we again observe clear evidence of manipulability: average prices in the markets in which we had bet `yes' are always higher than average prices in the markets in which we had bet `no'.

Taken together, our results imply that prediction markets \textit{can} be manipulated --- although the extent to which any manipulation persists depends on market characteristics in systematic ways. This result is suggested not just by theory (e.g. due to learning effects), but also appears in every sub-sample of the experimental data that we investigate. The success of our manipulations may be due to their relatively small size: while a shock of 5 percentage points is large enough to be statistically detectable, it appears to have been small enough to result in plausible looking market probabilities. In contrast, very large shocks may be insufficiently `persuasive' to persistently move market prices.

Our paper builds on several literatures that study the manipulability of prediction markets. An important precursor to our work is \cite{camerer1998can}, who studies the impact of making random bets on particular horses in the context of pari-mutuel racetrack betting.\footnote{See also \cite{brown2017anchoring} for a related experiment on `anchoring'.} Since \cite{camerer1998can}'s setting is not a prediction market, it differs from ours in multiple ways: for example, \cite{camerer1998can} studies a very different population of traders operating in an environment with a different pricing rule and a substantially more rapid resolution of uncertainty.\footnote{In addition, since horse races are repeated, it is plausible \cite{camerer1998can}'s bettors broadly agreed on what the prices of different horses ought to be. As our model makes clear, this makes manipulating the market price substantially more difficult and may help explain the contrast between our results and those that \cite{camerer1998can} obtains.} In addition, unlike in our experiment, \cite{camerer1998can} cancels his bets around 12 minutes after they have been made. For these reasons, it is perhaps unsurprising that we obtain very different results: while our trades have persistent effects even 60 days after they are made, the effect of \cite{camerer1998can}'s bets appears to vanish within around 20 minutes.

Somewhat closer to our study, a small literature attempts to trace out the impact of historical manipulation attempts within prediction markets. For example, \cite{rhode2004historical, rhode2006manipulating} study attempts to manipulate US prediction markets in the early 20th century; while \cite{hansen2004manipulation} examine an attempt to manipulate prediction markets during the 1999 Berlin state elections.\footnote{See also \cite{rothschild2016trading} for an examination of an apparent manipulation attempt during the 2012 US President Election.} Such studies clearly establish that many have \textit{attempted} to manipulate prediction markets. However, it is difficult to assess whether such attempts were successful using this approach since it is difficult to estimate the path that prices would have taken in the absence of the manipulation attempt. Our approach entirely side-steps this issue through the use of a randomised field experiment.

The manipulability of prediction markets has also been studied by a series of laboratory experiments: see \cite{hanson2006information, oprea2008can, veiga2009price, veiga2010information, jian2012aggregation, deck2013affecting, buckley2017effect} and \cite{choo2022manipulation}.\footnote{For laboratory evidence on prediction markets that does not study the question of manipulability, see also \cite{filippin2023risk}, \cite{mantovani2024prediction} and \cite{galanis2024information}.} However, the prediction markets created in the laboratory are typically quite different from those observed in the field.\footnote{For instance, they are populated by a different group of traders (typically undergraduate students) who form beliefs based on quite stylised information about abstract topics (e.g. the number of balls in an urn).} Moreover, due to the logistical difficulty of constructing markets in the laboratory, it is typically only feasible to run at most a handful of markets in the study. This makes identifying the impact of a manipulation difficult (especially given that observations within a market cannot be treated as independent) and makes the kind of heterogeneity analysis that we conduct impossible. For these reasons, it is useful to complement the existing laboratory evidence with evidence from the field.

Perhaps most relevant for our purposes, \cite{rhode2006manipulating} report the results of an experiment on the Iowa Electronic Markets that involves making 15 bets in total on 2 (inter-related) markets. Since \cite{rhode2006manipulating} were interested in manipulating a particular market (on the 2000 US Presidential Election), their experiment is very different from the large-scale and across-market field experiment whose results we report here. For this reason, our study is better powered to detect the long-run effects of manipulative trades.

Finally, although our paper contributes to the empirical literature on prediction markets, we also contribute to the theoretical literature. In particular, we adapt `price theory' style models of prediction markets \citep{gjerstad2005risk, manski2006interpreting, wolfers2006interpreting} so that they can cast light on price dynamics following a manipulation attempt.\footnote{For models in the `rational expectations' tradition (from which we depart), see also \cite{allengale, kumar1992futures, hanson2009manipulator, ottaviani2007outcome}.} To do this, we relax the assumption that agents are price takers and that prices are determined by competitive equilibrium; instead, agents in our model understand the extent to which their trades will alter the market price. Our model is, to our knowledge, the first attempt to theoretically trace out the path of prices following a price manipulation attempt.

The remainder of this article is structured as follows. Section \ref{sec_theory} presents our theoretical model; Section \ref{sec_background} describes the platform that we study experimentally; and \ref{sec_design} presents our experimental design. Our main results are contained within Section \ref{sec_results}; while Section \ref{sec_sweepcash} presents the results of our follow-up experiment. Section \ref{sec_conclusion} concludes with a discussion of the open questions raised by this research.

\section{Manipulation in theory}\label{sec_theory}

In this section, we develop some theoretical expectations regarding the impact of a manipulation attempt. As discussed earlier, our model builds on seminal work by \cite{manski2006interpreting}, \cite{gjerstad2005risk} and \cite{wolfers2006interpreting}. In contrast to these papers, prices in our model are determined by a common pricing algorithm instead of competitive equilibrium, which allows us to study the adjustment path of prices following a shock. In addition, our agents understand that their trades impact the market price and take this into account when choosing how to trade.

\textbf{Market.} We consider a single market in which traders can bet on whether an event will or will not take place. Traders can buy and sell \textit{yes shares}, each of which pays out 1 currency unit if and only if the event takes place. Likewise, traders can buy and sell \textit{no shares}, each of which pays out 1 currency unit if and only if the event does not take place. As pointed out by \cite{gjerstad2005risk}, one can assume without loss of generality that traders buy either yes shares or no shares (but not both).
%IR: I think this is oversimplified since he considers a model with linear pricing. But the idea is fine (e.g. since Manifold imposes this constraint).

\textbf{Traders.} There are a finite number of traders, indexed by $i \in \{0, 1, ..., m\}$. Each trader believes that the event will occur with probability $\pi_i \in [0, 1]$; thus, the prior beliefs of traders can differ. For simplicity, we assume that each trader has the same initial wealth $w > 0$ and vNM utility function $u \colon \mathbb{R}^+ \mapsto \mathbb{R}$. Let $q_{y}^i \geq 0$ and $q_{n}^i \geq 0$ denote trader $i$'s purchase of yes and no shares respectively (as discussed before, either $q_{y}^i = 0$ or $q_{n}^i = 0$). We assume that each trader maximises expected utility, i.e.
\begin{equation}\mathbb{E}[u(q_{y}^i, q_{n}^i)] =
\begin{cases}
\pi_i u(w + q_y^i - C(q_y^i)) + (1 - \pi_i)u(w - C(q_y^i)) & \text{if } q_y^i \geq 0\\
\pi_i u(w - C(q_n^i)) + (1 - \pi_i)u(w + q_n^i - C(q_n^i)) & \text{if } q_n^i \geq 0 \end{cases}
\end{equation}
where $C(q_y^i)$ and $C(q_n^i)$ are the costs of purchasing yes and no shares. We assume that $u$ is twice differentiable with $u'(w_s) > 0$ and $u''(w_s) < 0$ for all $w_s > 0$, where $w_s$ is the agent's wealth in a generic state. To prevent corner solutions, we assume that $\lim_{w_s \rightarrow 0^+}u'(w_s) = \infty$. Finally, we assume that $u$ exhibits decreasing absolute risk aversion, i.e. that $-u''(w_s)/u'(w_s)$ is strictly decreasing in $w_s$ for all $w_s > 0$. Although this assumption is unnecessarily strong, it simplifies the analysis considerably; see the discussion below.
%IR: shall we mention that $u$ is continuous (even at q = 0)?

\textbf{Pricing.} We assume that costs are determined by the constant product rule; this pricing algorithm underpins the prediction market that we study experimentally (see Section \ref{sec_background}). Under this rule, prices are set by an automated market maker (AMM) that holds reserves of yes and no shares, denoted ($y$, $n$). To illustrate the mechanics of the rule, suppose that a trader wishes to spend $x$ currency units to purchase yes shares (analogous comments apply to purchases of no shares). The trader first transfers $x$ to the AMM, which converts this into $x$ yes shares and $x$ no shares; thus, the AMM's reserves become ($y + x$, $n + x$). The AMM then transfers the trader the number of shares $q$ that satisfies
\begin{equation}\label{eq_constantproduct}
(y + x - q)(n + x) = yn
\end{equation}
That is, it supplies the number of yes shares necessary to restore the product of its reserves to its previous value $yn$. This rule implicitly defines the number of shares that a trader receives $q$ given the amount of currency they have spent $x$. Indeed, solving (\ref{eq_constantproduct}) for $q$ yields
\begin{equation}\label{eq_q}
q = \frac{x (n + x + y)}{n + x}
\end{equation}
Equivalently, after letting $C(q) = x$ denote the total cost of the shares, one sees that
\begin{equation}
C(q) =  \frac{\sqrt{(n  - q + y)^2 + 4 n q} + q - n - y}{2}
\end{equation}
From this, it is straightforward to derive the marginal and average cost functions. The following lemma characterises their key properties.\footnote{All proofs are collected in Appendix \ref{sec_proofs}.}

\begin{lemma}\label{lemma1}
Under the constant product rule, the marginal cost of yes shares $MC(q)$ satisfies \(\mathrm{(i)}\) $MC(0) = n/(n + y)$, \(\mathrm{(ii)}\) $MC'(q) > 0$ for all $q \geq 0$, and \(\mathrm{(iii)}\) $\lim_{q \rightarrow \infty} MC(q) = 1$. Similarly, the average cost of yes shares $AC(q)$ satisfies \(\mathrm{(i)}\) $\lim_{q \rightarrow 0^+} AC(q) = n/(n + y)$, \(\mathrm{(ii)}\) $AC'(q) > 0$ for all $q > 0$, and \(\mathrm{(iii)}\) $\lim_{q \rightarrow \infty} AC(q) = 1$.
\end{lemma}

As Lemma \ref{lemma1} emphasises, pricing under the constant product rule is non-linear. For a small purchase of yes shares, both the marginal and average costs are close to $n/(n + y)$ and thus determined by the ratio of yes and no shares held by the AMM (see Figure \ref{fig:mc_ac}). As the size of the purchase rises, yes shares become more scarce and so both average and marginal costs rise. In the limit as $q \rightarrow \infty$, average and marginal costs converge to $1$; notice that such a purchase is necessarily unattractive since the expected value of each share cannot exceed $1$.

\textbf{Optimal trading.} We now characterise the optimal trading behaviour of the market participants. As the next result shows, the `marginal price' $n/(n + y)$ plays a key role in pinning down which types of trade participants will make. Specifically, participants with optimistic beliefs ($\pi_i > n/(n + y)$) will bet that the event will take place; whereas participants with pessimistic beliefs ($\pi_i < n/(n + y)$) will bet that event will not take place.

\begin{lemma}\label{lemma2}
Define $p = n/(n+y)$. Then
\begin{itemize}
    \item If $\pi_i > p$, the trader will buy a positive quantity of yes shares.
    \item If $\pi_i = p$, the trader will not hold any shares.
    \item If $\pi_i < p$, the trader will buy a positive quantity of no shares.
\end{itemize}
\end{lemma}

At a technical level, Lemma $\ref{lemma2}$ extends \cite{arrow1965theory}'s analysis of optimal betting to the case of non-linear pricing. As pointed out by \cite{arrow1965theory}, even risk-averse traders will choose to make a positive bet if this has positive expected value since they are locally risk neutral around $w_s = 0$. Since pricing is locally linear around $q = 0$, it is unsurprising that this argument extends to our case and that the `marginal price' at zero $n/(n + y)$ plays a key role in the argument.

Next, we examine how the behaviour of traders responds to an exogenous shock to the market price. More precisely, we suppose that an external participant makes a purchase of yes shares, thereby increasing $n$, decreasing $y$ and pushing up the marginal price $n/(n + y)$. The next result establishes how traders react qualitatively to the shock.

\begin{lemma}\label{lemma3}
Suppose that the marginal price increases from $p$ to $p + \Delta$. Then
\begin{itemize}
    \item Traders with $\pi_i \geq p + \Delta$ will decrease their holdings of yes shares.
    \item Traders with $\pi_i \in (p, p + \Delta)$ will switch from holding yes shares to holding no shares.
    \item Traders with $\pi_i \leq p$ will increase their holdings of no shares.
\end{itemize}
\end{lemma}
%IR: does this break the flow? Should we delete FOC paragraph?
To sketch the proof, consider the case of $\pi_i \geq p + \Delta$. Ignoring the cases of $\pi_i = 1$ and $\pi_i = p + \Delta$ (which require a separate argument), one first shows that optimal purchases must satisfy a first order condition
\begin{equation}\label{eq_foc2}
\frac{\pi_i u'(w + q - C(q))}{(1 - \pi_i)u'(w - C(q))} = \frac{C'(q)}{1 - C'(q)}
\end{equation}
For any fixed $q > 0$, the shock increases marginal costs, thus shifting the right-hand side of (\ref{eq_foc2}) upwards. Meanwhile, under the assumption of decreasing absolute risk aversion, one can show that the shock shifts the left-hand side (i.e. the ratio of marginal utilities) downwards. Each of these forces separately decreases the optimal purchase, which means that traders will unambiguously choose to hold fewer yes shares.\footnote{As this discussion makes clear, the assumption of decreasing absolute risk aversion is unnecessarily strong. Without it, however, one would need to assume an appropriate bound on the income effect.}

More substantively, Lemma \ref{lemma3} underscores the importance of reversion. It is easy to check that traders who decrease their holdings of yes shares will put downward pressure on the market price. Likewise, the market price will fall if a trader increases their holdings of no shares. Thus, the model predicts that, following an exogenous increase in the price, the behavioural adjustment of an arbitrarily chosen trader will lead the price to revert towards its original value.

Although Lemma \ref{lemma3} captures the logic of reversion, it does not settle the quantitative question of how much prices should revert following a manipulation attempt. In particular, it does not tell us whether reversion should be complete. In addition, Lemma \ref{lemma3} does not describe the speed of reversion and the shape of the price path following the manipulation. To shed light on these more difficult questions, we now turn to simulations.

\textbf{Simulations.} To simulate the market under various conditions, we assume that $u(w_s) = \ln(w_s)$; note that this functional form satisfies our maintained assumptions on $u$. One advantage of the logarithmic specification is that it allows the optimal trades to be solved explicitly. For example, if $\pi_i > n/(n + y)$, then the optimal spend on yes shares is
\begin{equation}
x^* = \frac{n \left(\sqrt{y [(n+4 (1 - \pi_i) \pi_i w)y +4 (1 - \pi_i) \pi_i w (n+w)]}-2 (1 - \pi_i) w-y\right)}{2 (1 - \pi_i) (w+y)}
\end{equation}
and the optimal purchase $q$ may be computed using \eqref{eq_q}. Since the simulations involve solving well over $10^7$ optimisation problems (see below), avoiding numerical optimisation delivers an important reduction in speed.

In all simulations, the market is initialised and given $t$ periods to reach a stable state. As in our experiments, an external participant then purchases the number of yes shares required to increase the price by $0.05$; the market is then given $t'$ periods to adjust. At each time, one trader is randomly selected to re-adjust her holdings. This introduces randomness into the process; thus, we run each simulation 10,000 times and report the average of the realised prices at every time.\footnote{All simulations can be reproduced using \href{https://github.com/Itzhak95/prediction_markets_model}{this code}.}

In the \textit{baseline case}, we assume that $m = 10$ (i.e. that there are 11 traders), that each trader has wealth $w = 100$, and that the initial reserves are $(n, y) = (1000, 1000)$. We also assume that trader $i$ has a belief $\pi_i = i/m$ for $i \in \{0, 1, ..., m\}$ (i.e. that beliefs are uniform). Finally, we set $t = t' = 100$; this turns out to give the market prices more than enough time to stabilise following the shock. The results are shown in Panel A of Figure \ref{fig:res_sim}. As can be seen, the average price is stable at $0.5$ before the shock; this is unsurprising given that the initial reserves are balanced ($n = y$) and that beliefs are symmetric around $0.5$. After the shock increases the price to $0.55$, the behavioural responses of the agents induce the price to start moving back towards its original value of $0.5$. However, reversion takes time; and even once prices have stabilised, they have only reverted by around 40\%.

It is natural to wonder how these results depend on the assumed market conditions; we thus now report some variations. First, we vary the number of traders, considering $m \in \{10, 20, ..., 60\}$. Second, following \cite{manski2006interpreting}, we allow traders to revise their beliefs in light of the market price $p$. Specifically, we suppose that posterior beliefs are given by $\pi_i' = \lambda p + (1 - \lambda)\pi_i$,  where $\lambda \in [0, 1]$ is the `learning rate'; we consider $\lambda \in \{0, 0.2, ..., 1\}$. Finally, we vary the degree of prior agreement amongst the traders. Specifically, we assume that trader $i$ has a belief $\pi_i = \alpha (i/m) + (1 - \alpha)0.5$ where $\alpha \in [0, 1]$; note that that this nests the baseline case of uniform beliefs $(\alpha = 1)$ along with the `opposite' case of a common prior belief of $0.5$ ($\alpha = 0$). We run separate sets of simulations for all $\alpha \in \{0, 0.2, ..., 1\}$.

Figure \ref{fig:res_sim} displays the most `extreme' version of each variation; see also Table \ref{tab_full_sim} for a more detailed analysis that describes the degree of reversion under all the parameter combinations that we consider. Several results are apparent. First, increasing the number of traders $m$ increases the speed of reversion in the short run as well as increasing the total magnitude of reversion once prices have stabilised. This is intuitive: increasing the number of traders increases the amount of wealth at their disposal, thereby making it easier for them to undo the impact of the manipulative trade. Second, increasing the learning rate $\lambda$ makes the manipulation more `sticky'; indeed, in the extreme case of $\lambda = 1$, no reversion is observed whatsoever. Intuitively, increasing $\lambda$ means that the manipulative trade is `more persuasive' (i.e. leads to a larger update to trader beliefs) and thus has a larger long-run effect on the market price. Finally, increasing the amount of prior agreement amongst traders makes the market harder to manipulate. In the extreme case where all traders have a common belief of $0.5$, moving the price away from $0.5$ is especially difficult.

\begin{figure}[t!]
\centering
\caption{Price adjustment under different market conditions}
\label{fig:res_sim}
\begin{subfigure}[t]{0.49\linewidth} % Changed [b] to [t] for top alignment
    \centering
    \includegraphics[width=\linewidth]{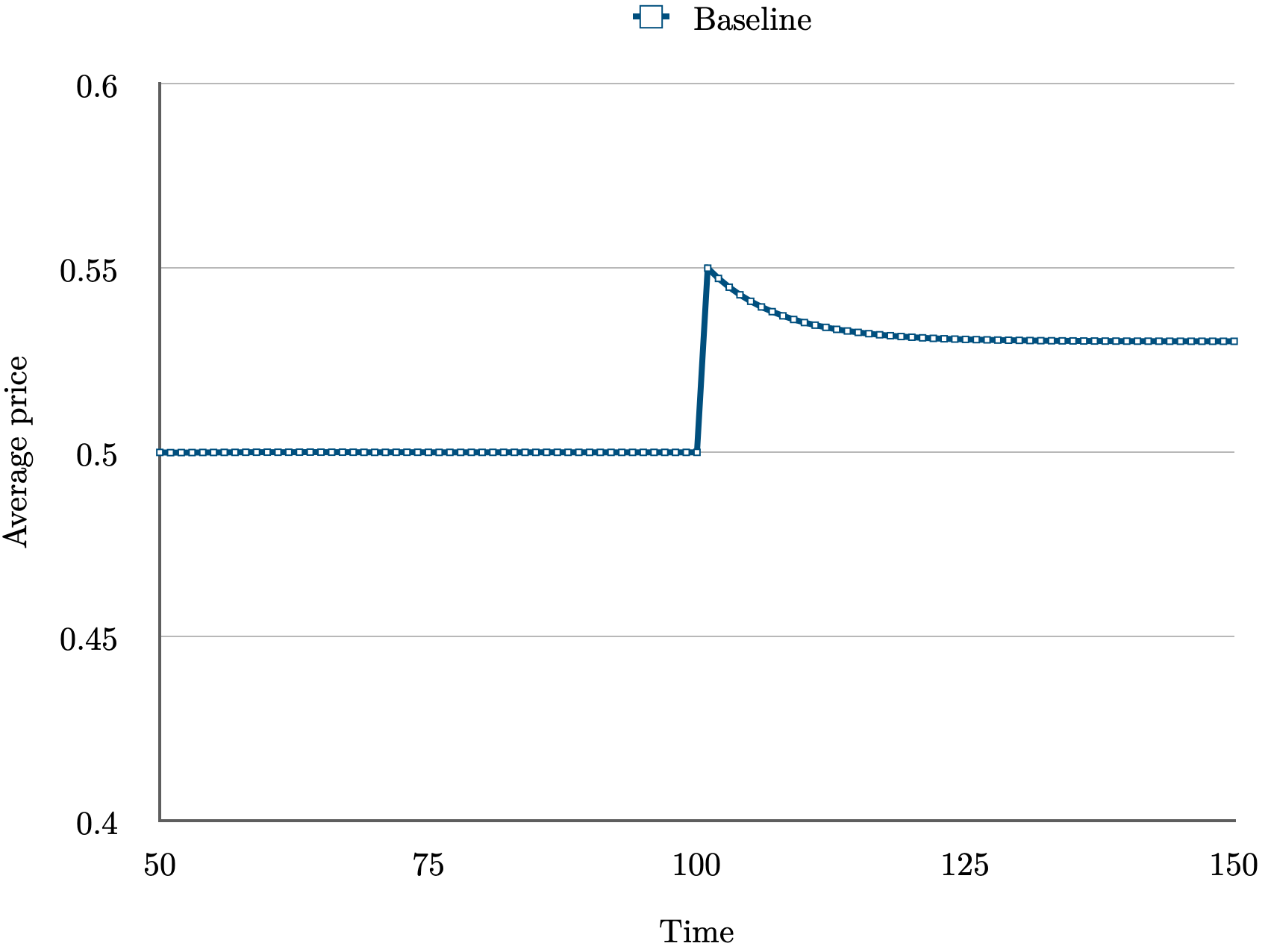}
\caption{Baseline}\label{panel_a}
\end{subfigure}
\hfill
\begin{subfigure}[t]{0.5\linewidth} % Changed [b] to [t] for top alignment
    \centering
    \includegraphics[width=\linewidth, trim=0cm 0.5cm 0cm 0cm, clip]{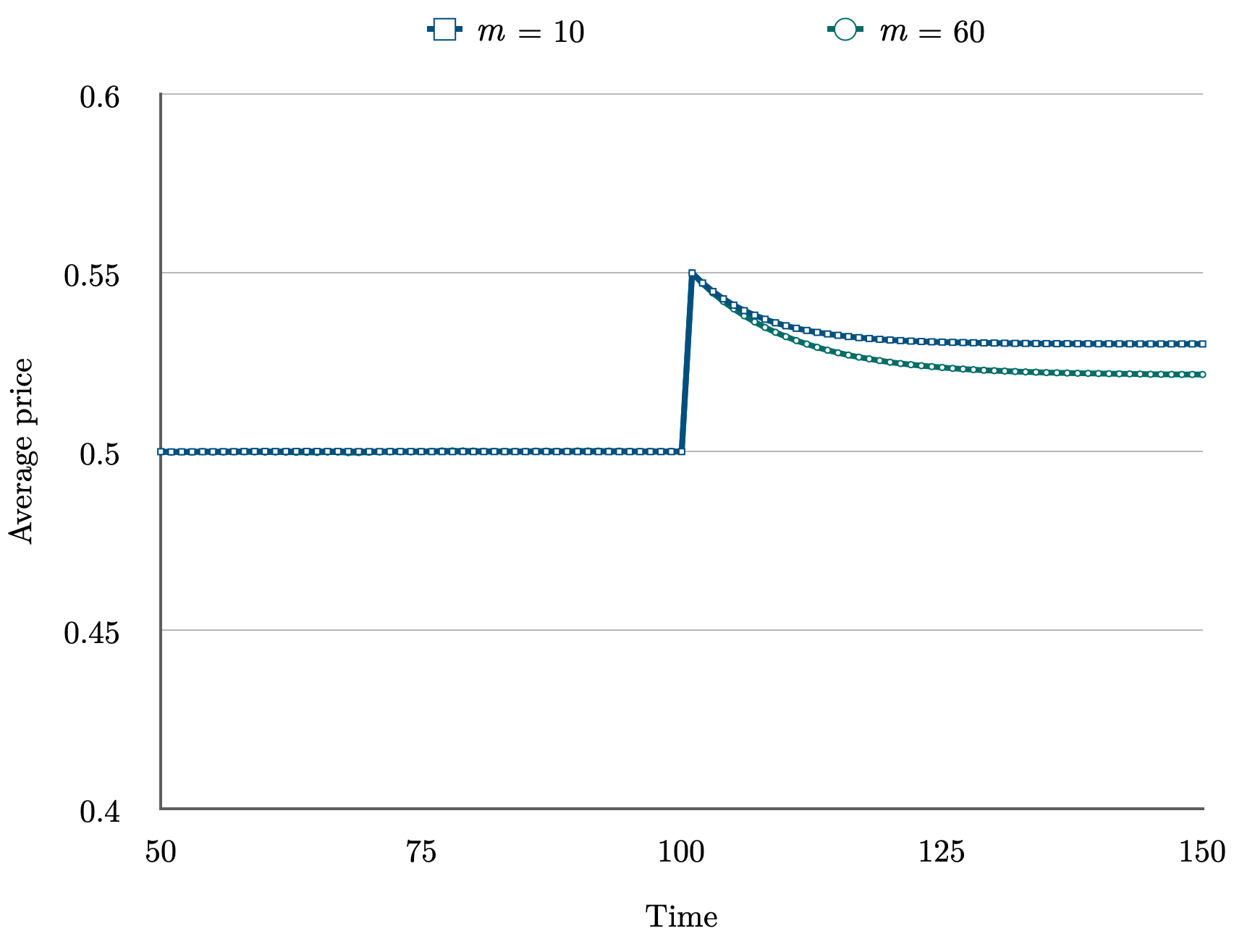}
    \caption{Varying the number of traders}
\end{subfigure}
\vspace{-.5cm} % Adjust this space to bring the top and bottom rows closer if needed
\begin{subfigure}[t]{0.49\linewidth} % Changed [b] to [t] for top alignment
    \centering
    \includegraphics[width=\linewidth]{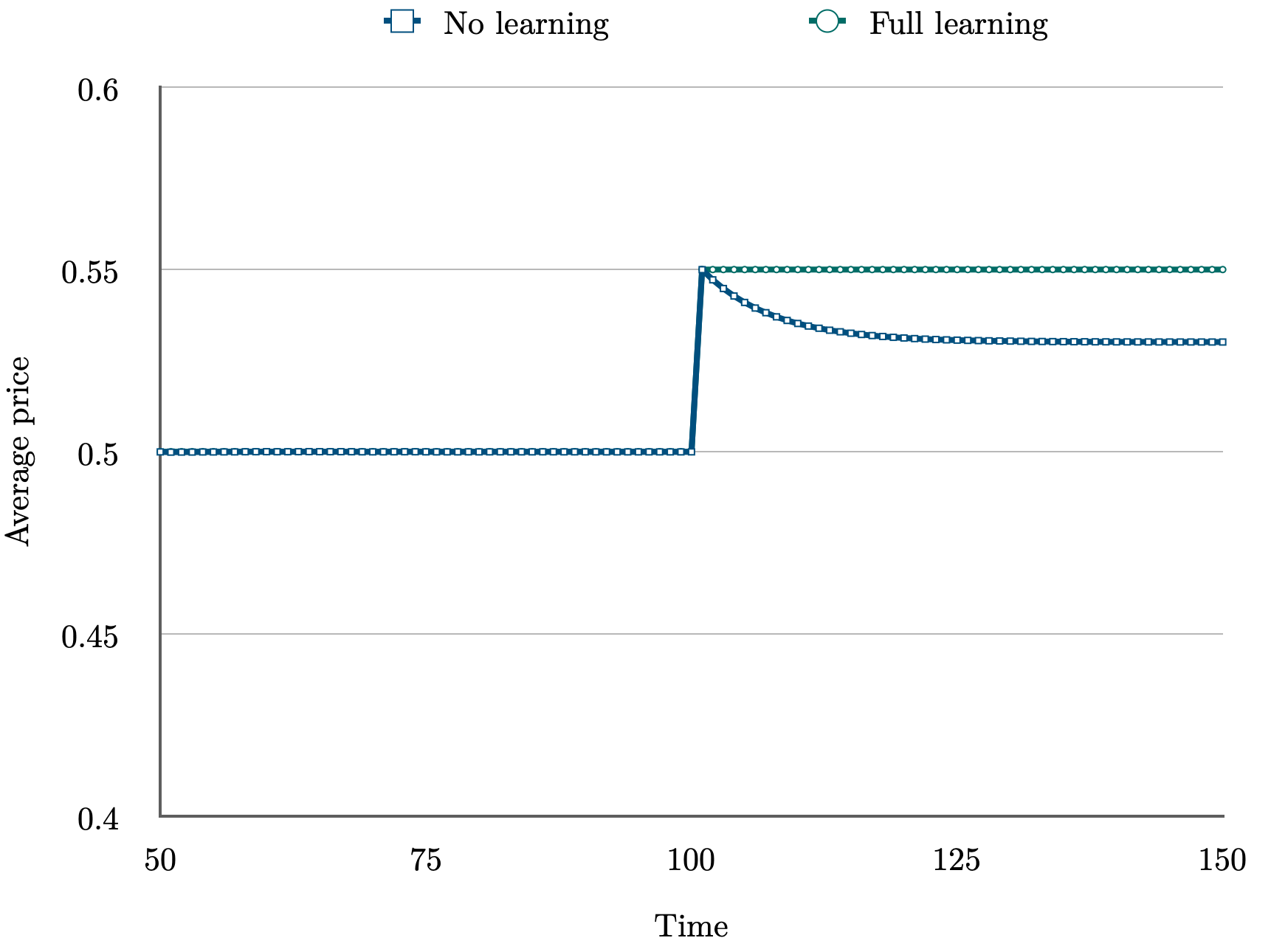}
    \caption{Varying the learning rate}
\end{subfigure}
\hfill % Adjust or remove to manage spacing
\begin{subfigure}[t]{0.5\linewidth} % Changed [b] to [t] for top alignment
    \centering
    \includegraphics[width=\linewidth, trim=0cm 0.2cm 0cm 0cm, clip]{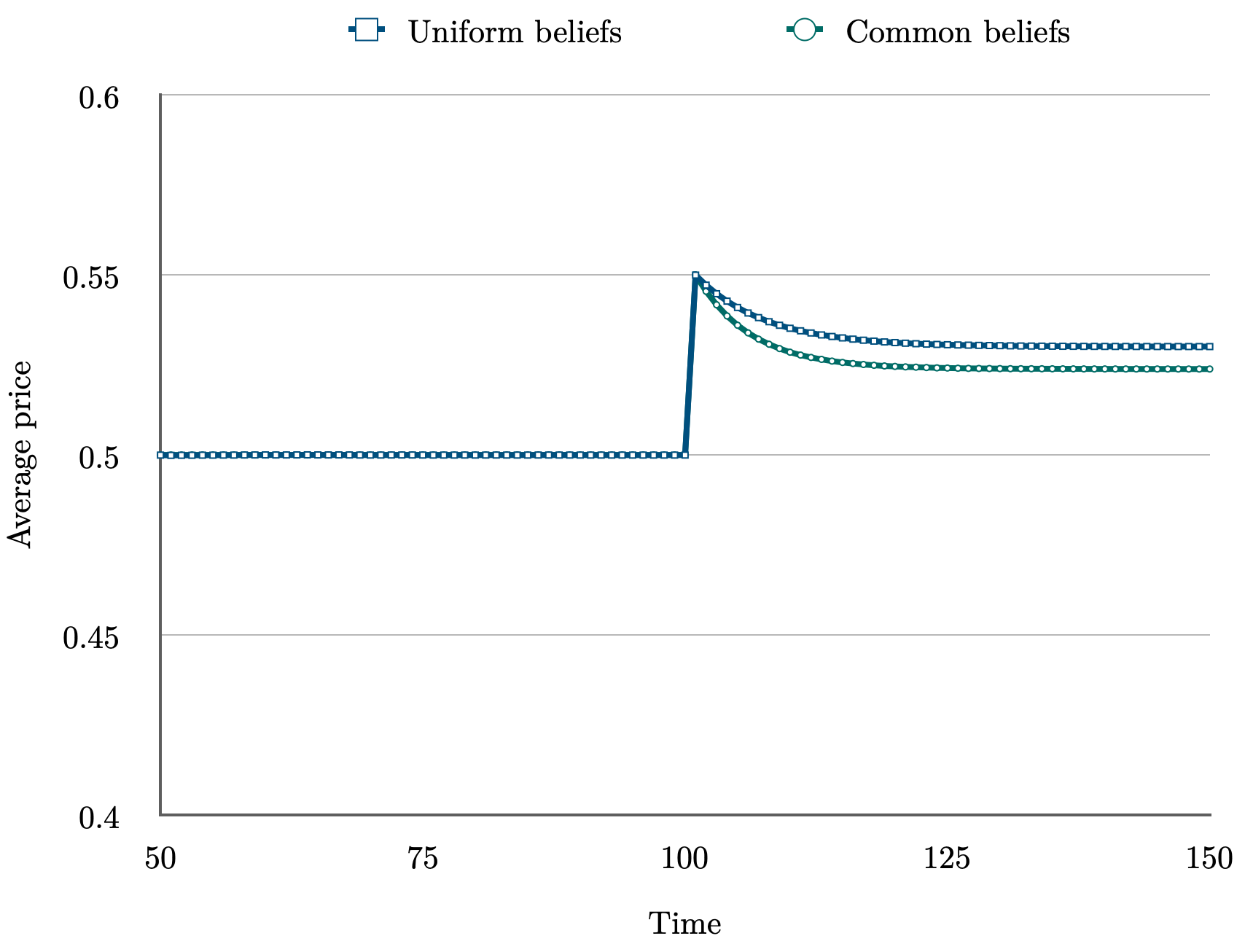}
    \caption{Varying prior agreement}
\end{subfigure}
\begin{minipage}{\linewidth}
\vspace{0.8cm}
   \footnotesize \textit{Notes}. This figure plots the average price (across 10,000 simulations) from time $t = 50$ to $t = 150$ under different market conditions. Panel (a) displays the baseline simulations; Panel (b) compares $m = 10$ with $m = 60$; Panel (c) compares $\lambda = 0$ and $\lambda = 1$; and Panel (d) compares uniform beliefs with a common belief.
\end{minipage}
\end{figure}

Taken together, the model outlined in this section generates three key predictions concerning the impact of price manipulation by an external market participant. First, the model predicts that manipulation can systematically affect the market price --- both in the short run (due to the need for behavioural adjustment), but also in the long-run (e.g., due to `learning effects'). Second, however, the model also predicts that the effect of manipulation should be somewhat `undone' by future trades; this dynamic is captured qualitatively by Lemma \ref{lemma3} and quantified by the simulation results. Third, the model suggests that the degree of reversion should vary by market type. In particular, markets with more traders, a lower degree of learning, and more prior agreement should be harder to manipulate. We will examine the extent to which these predictions are borne out by the data in the subsequent sections.

\section{Institutional background}\label{sec_background}

\textbf{Manifold markets.} We conducted our experiment on the \href{https://manifold.markets/}{Manifold Markets} platform, which was founded in 2021. During our experiment, Manifold was the largest prediction market website in the world as measured by the total number of markets hosted on the platform.\footnote{As of August 16 2024, there were 10 markets on PredictIt, 312 markets on Kalshi, 504 markets on Polymarket, and well over 2,000 markets on Manifold.} In addition, Manifold was one of the largest prediction market platforms as measured by number of users, with around 10,000 `active' users at the start of the experiment in December 2023 \citep{manifold_stats}.\footnote{According to Similarweb estimates of website traffic, Manifold was the second most viewed prediction market website in April 2024. Specifically, the page view estimates are 1,145,000 for PredictIt, 673,000 for Manifold, 531,000 for Polymarket, and 177,000 for Kalshi.} The large number of markets on Manifold was crucial to the design of our large-scale field experiment; such an experiment would not have been possible using more traditional platforms (e.g. the IEM) that have previously been studied in the academic literature \citep{wolfers2004prediction}.
%IR could plot correlations between Manifold and other sites?

\textbf{Trading.} Figure~\ref{fig:trump_mkt} provides an example of a particular market on the platform. Traders can view the full history of market prices and are given the opportunity to purchase yes or no shares. If they choose to make a trade, they are shown how much this will cost and how much it will move the (marginal) market price. Each market is equipped with a closing date, which usually corresponds to the time at which the relevant question will be resolved. Below the price chart, traders can read a description of the exact resolution criteria as well as any comments that have been made by the traders over the course of the market's lifespan.

\textbf{Pricing.} Trading within each market is conducted via an AMM that implements Maniswap, a generalisation of the constant product rule that was studied in Section \ref{sec_theory}. Under Maniswap, the AMM adjusts its reserves $(y, n)$ in order to hold the expression $y^p n^{1 - p}$ constant. The parameter $p \in (0, 1)$ allows markets to be initialised at prices other than 50\% without `wasting' shares (see \citealp{maniswap_24} for the details). Notice that, in the common case where a market is initialised at 50\% (corresponding to $p = 0.5$), the Maniswap algorithm exactly reduces to the constant product rule. Moreover, one can verify that an (appropriately modified) version of Lemma \ref{lemma1} holds under Maniswap, which allows one to establish analogous versions of Lemmas \ref{lemma2} and \ref{lemma3} in this generalised environment.\footnote{More precisely, all results continue to hold after one has changed the `marginal price' from $n/(n + y)$ to $n p/(n p - p y + y)$.}

\textbf{Users.} According to an informal survey, which should be treated with some caution due to the possibility of self-selection, Manifold users appear to skew heavily male \citep{plasma_survey_2024}. The majority of users reside in the United States; and these users are substantially more likely to vote for the Democratic than the Republican party. Most users stated that they were either aligned with or interested in the `rationalist' movement; interest in and alignment with the `effective altruism' movement are also common.

\textbf{Innovations.} In some respects, Manifold is an unusual platform. First, most markets are user-created and user-resolved. This allows Manifold to host a large number of markets, which is crucial for the feasibility of our large-scale experiment. Second, Manifold provides rich data on each market's characteristics, which facilitates our heterogeneity analysis. Third, a large portion of trade is conducted by `bots' (i.e. automated trading algorithms) that attempt to exploit inefficiencies in the markets. Fourth, markets are run on a platform-specific currency (`Mana'). While Mana (M) can be purchased using real currency (at a rate of \$1 = 100M during the experiment), it is not possible to convert Mana back into dollars.

\textbf{Incentives.} Despite this latter feature, there are several incentives to earn Mana through trading. First, during the experiment, Mana could be converted into donations to a charity of the user's choice at the same \$1 = 100M rate; in total, Manifold users had raised \$316,000 for charity through this mechanism as of 16 May 2024. This incentive is plausibly especially effective given that many Manifold users are interested in `effective altruism' and thus may view money and charitable donations as fungible. Second, Manifold ranks users based on their monthly earnings: this may activate social-image incentives. Third, Manifold provides users with personalised information about their trading performance and predictive accuracy through Brier scores and calibration charts; this information may also tap into self-image incentives \citep{benabou2006incentives}. Taken together, these incentives help explain why Manifold had over 10,000 active users during the experiment despite running on a platform-specific currency. As one especially successful trader on Manifold put it \citep{interview},

\begin{quote}
In the unusual world in which I find myself, for better or worse, doing well on a prediction markets website is somewhat of a badge of honour . . . I wish I had more noble motivations but, alas, I think that’s a good chunk of it . . . Another important motivation for me using Manifold relates to charitable giving.
\end{quote}

The importance of the financial incentives is also underscored by the `devaluation' of Mana that took place in April 2024. On 23 April, it was announced that Manifold would reduce the rate at which Mana could be converted to charitable donations by a factor of 10. However, users wishing to avoid the devaluation were given a one week `grace period' in which they could donate their earnings at the old (and substantially better) rate.\footnote{This grace period was ultimately extended to 15 May 2024.} Figure~\ref{fig:dev} shows the impact of the announced devaluation on donations to GiveWell's maximum impact fund, which is the charity that received the largest value of donations from Manifold users. As can be seen, the number of donations increased by over a factor of 100 in the week following the announcement; whereas the value of donations increased by over a factor of 200. This suggests that, although Manifold users may be motivated by social and self-image incentives, they are also responsive to the financial incentives offered by the platform.

\begin{figure}[H]
    \centering
    \caption{Calibration\vspace{-0.5em}}\label{fig:calibration}
    \includegraphics[width=0.8\textwidth, keepaspectratio]{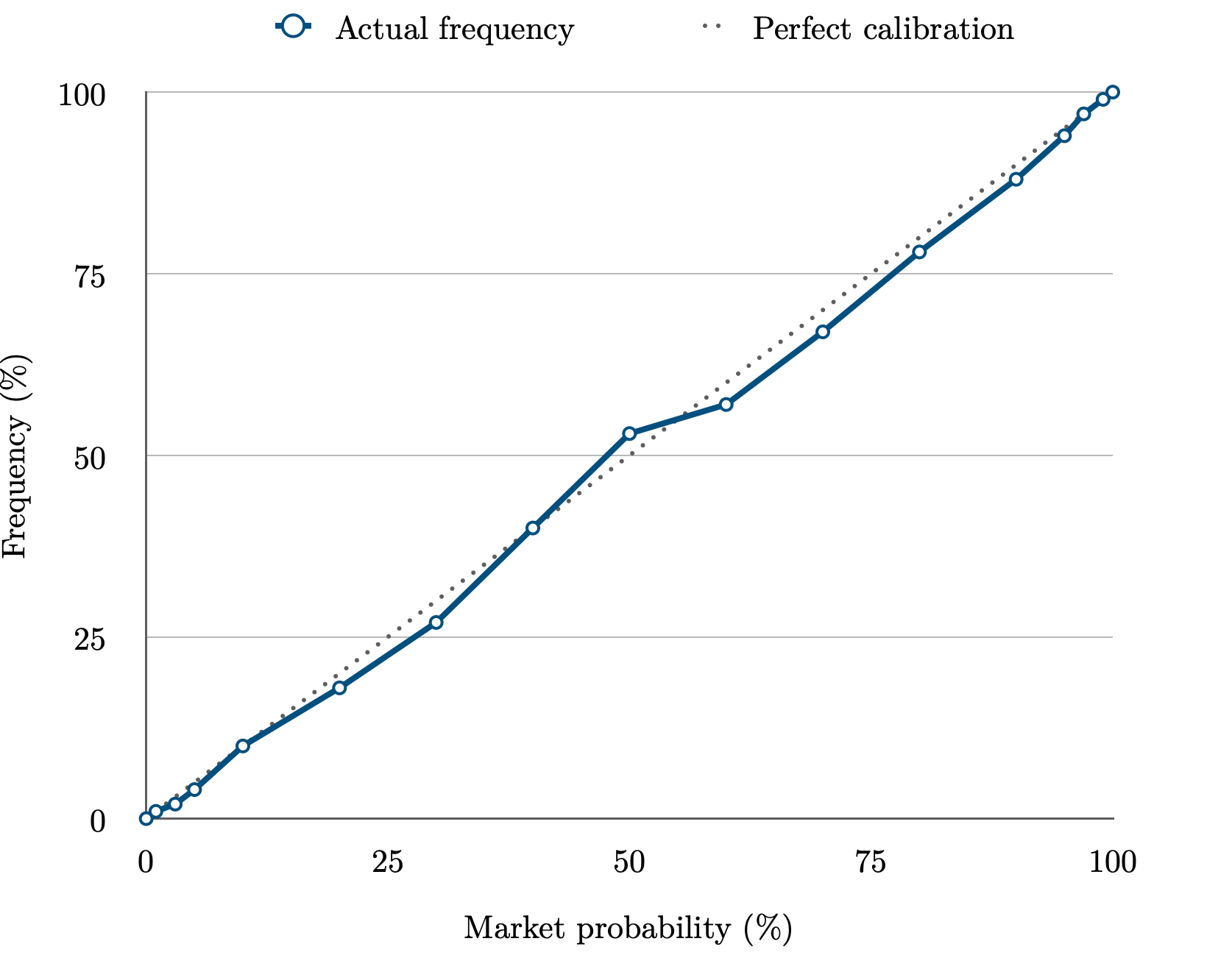}
    \begin{minipage}{10.5cm}
    \vspace{0.3cm}
    \footnotesize \textit{Notes}. This figure shows calibration on Manifold as of 20 November 2023, i.e. the relationship between the market price and the probability that a market resolves `yes'. Source: \cite{calibration}.
\end{minipage}
\end{figure}

\textbf{Predictive performance.} Given these incentives, it is not surprising that the predictive performance of Manifold is comparable to that of more traditional platforms. First, as shown by Figure \ref{fig:calibration}, the markets are generally well-calibrated: for example, markets on Manifold with prices close to 40\% have historically resolved `yes' around 40\% of the time. Second, in a study of the 2022 US midterm elections, Manifold outperformed the more traditional prediction markets in the sample \citep{First2024}. Third, Manifold achieves Brier scores that are comparable but slightly worse than Metaculus \citep{EAForum2024}. Thus, based on existing evidence, it appears that Manifold markets operate in broadly comparable ways to more traditional prediction markets. This is unsurprising in light of previous work on the role of money in prediction markets. As \cite{servan2004prediction} find, `the essential ingredient' for prediction market accuracy `seems to be a motivated and knowledgeable community of traders'---which Manifold certainly appears to have---as opposed to whether the traders are motivated purely by financial incentives.

\section{Experimental design}\label{sec_design}

\textbf{Intervention.} We now outline the structure of our experiment. As discussed earlier, the basic idea was to randomly shock the prices of a large number of prediction markets. More specifically, we either purchased yes shares until the price rose by 5 percentage points (the `yes' group), purchased no shares until the price fell by 5 percentage points (the `no' group), or did nothing (the control group). Each given market was equally likely to be assigned to any of these groups. We then observed the prices in every market in order to study whether the effects of our shocks fade over time.

As explained in the pre-registration \citep{pre-reg}, we focus on the difference in mean prices between the `yes' and `no' groups since this is our best-powered comparison. Two points should be noted. First, since the treatment is randomly assigned, the difference in average prices between these groups has a causal interpretation. Second, if prices in the markets can be approximated by a random walk, then one would expect the variance of the prices (and thus our standard errors) to increase over time. For this reason, we focus on effects within 7 days of our intervention; although we also conduct analyses with longer time horizons.

\textbf{Exclusion criteria.}  We only bet on binary markets, i.e. markets that must resolve as `yes' or `no' (or N/A). We excluded various types of market from our sample. First, we excluded markets that do not resolve based on an external event by the end of 2025. While very long-run markets (e.g. resolving in 2030) may be relatively easy to manipulate, they are also less representative of the markets that operate on other platforms. Second, since very small markets are less likely to form the basis for influential probability estimates, we excluded any market that had fewer than 10 traders at the time of our bet. Third, to obtain a large sample with our experimental budget, we excluded any market which would have cost over 200M to move in either direction by 5 percentage points.\footnote{In practice, the effect of this restriction is to reduce the number of markets in our sample with initial prices close to 0\% to 100\%. To see why this is true, note that, under the constant product rule, it is infinitely expensive to increase the price from 95\% to 100\% (or to decrease the price from 5\% to 0\%).} Fourth, since markets can be volatile at the start of their life cycle (which increases standard errors), we excluded markets that started within the last 7 days; to satisfy our data collection requirements, we also excluded markets that ended within 30 days of our bet. Fifth, to reduce the scope for arbitrage and informational effects (discussed in detail later), we excluded any market that was closely related to another market that was already in our sample (e.g. “Will Trump win?” vs “Will Trump lose?”).

\textbf{Data.} We collected hourly price data, starting 24 hours before the bet and continuing for 30 days after the bet. Since we collected data on 817 markets, this means that we have $817 \times 24 \times 31 \approx 600,000$ price observations in total. We also collected rich data on each market's characteristics. First, we recorded the number of traders within each market at the time of our bet; our model suggests that this variable should influence the market's manipulability. Second, we recorded whether each market's question was also present on another prediction platform (Metaculus) that is frequently discussed on Manifold and sometimes used as the source for Manifold questions. If a question is on Metaculus, then one would expect that `price learning' should be lower and thus the market should be harder to manipulate (see Section \ref{sec_theory}). Third, we collected three measures of the `activity of a market' --- the total volume of trade, the volume of the trade within the last 24 hours, and the number of comments left by traders. Since our model identifies a `time period' with a trade, it suggests that more active markets should effectively run on a faster `clock speed' and thus revert more quickly. For completeness, we also recorded some less interesting variables that were available on the platform; see \cite{pre-reg} for a full description.

\textbf{Power.} Based on the calculations in \cite{pre-reg}, we needed $n = 849$ markets in order to have 90\% power to detect a difference of 3 percentage points between the average prices in the `yes' and `no' groups after 1 week.\footnote{When analysing the data, we realised that some of the $849$ markets on which we had bet inadvertently violated at least one of our exclusion criteria. After discarding these markets, we ended up with $n = 817$ markets in our sample. We should emphasise that our main results are entirely unaffected by the inclusion (or lack of inclusion) of these markets.} Since the observed reversion was ultimately substantially lower than we had anticipated, this sample size turned out to be unnecessarily large for this purpose. However, the size of our sample remains useful for our heterogeneity and longer-run analyses.

\textbf{Timelines}. We pre-registered our experiment (with an analysis plan) in December 2023.\footnote{In general, we conformed closely to the analyses that we had pre-registered. However, we ultimately deviated from the plan in a few (generally minor) ways. All deviations from the pre-registered plan can be viewed in our pre-registration; see the document `Analysis plan (updated after the experiment)'.} We made the bets over a 5 month period from December 2023 until April 2024. We concluded our main data collection in May 2024 (30 days after the final bet).

\section{Experimental results}\label{sec_results}

\subsection{Sample}

We now describe the results of the experiment, beginning with an overview of the markets in our sample. Table \ref{tab:desc_stat} reports some descriptive statistics. As can be seen, around 9\% of the market questions could also be found on the Metaculus platform. On average, users had left 3.5 comments at the time of our trade; however, this varies substantially across markets with a minimum of 0 and a maximum of 83. At the time of our trades, 27 traders had participated in each market on average. If one considers the markets that had resolved as of 10 July 2024, one sees that this number had risen to 47 traders on average by the end of the markets' lifespans. Meanwhile, the total number of trades had increased on average to 136, with a minimum of $0$ and a maximum of $1,400$.

\begin{table}[H]
\begin{threeparttable}
\caption{Descriptive statistics}\label{tab:desc_stat}
\vspace{-0.5em}
\begin{tabular}{lcccc}
\hline
Variable       & \hspace{1.5em}Mean\hspace{1.5em}     & Std. dev. & Minimum & Maximum   \\ \hline
24 hour volume & 24.7   & 119.9            & 0   & 1,701  \\
Total volume   & 1,879 & 4,279           & 89  & 56,485 \\
Metaculus      & 0.093    & 0.291              & 0   & 1     \\
Comments       & 3.53    & 7.26              & 0   & 83    \\
Traders        & 27.0   & 25.0             & 10  & 300   \\ \hline
Final traders  & 47.1   & 44.9             & 10  & 321   \\
Final trades   & 136.0  & 189.7            & 13  & 1,400  \\ \hline
\end{tabular}
\begin{tablenotes}
\footnotesize
\item \hspace{-0.2em}\textit{Notes}. This table shows descriptive statistics for the markets in our sample. The first 5 variables show statistics for the full sample ($n = 817$). The last two variables show statistics for the 127 markets that had resolved as of 10 July 2024.
\end{tablenotes}
\end{threeparttable}
\end{table}

Table \ref{tab:topics} provides an overview of the main topics addressed by the markets in our sample; see also Fig \ref{fig:cloud} for a visual illustration. As can be seen, markets on politics (especially the US 2024 Presidential election) and political conflict (especially the conflicts in Gaza and Ukraine) are common. In addition, many markets concern artificial intelligence and the speed of its development. Our sample also includes a variety of markets on sport, popular culture (e.g. the outcomes of Oscar nominations), and various macroeconomic indicators (e.g. inflation and interest rates).

\subsection{Main results}

\textbf{7 day effects.} We now examine the extent to which our manipulations had a persistent effect over time. To obtain some initial graphical evidence on this question, Figure \ref{fig:avg_prices} plots the average price in each of the three treatment groups, starting 24 hours before the trade and continuing for 7 days after the trade. As can be seen, the average prices in each group are stable in the 24 hours before the trade. By design, our intervention immediately increases the average price in the `yes' group by 5 percentage points and immediately decreases the average price in the `no' group by 5 percentage points, thus creating a gap of 10 percentage points in average prices between the groups.\footnote{Since our data is hourly, there is some scope for reversion even in the hour in which the trade is made: for this reason, the initially observed impacts of the yes and no shocks are slightly below 5 percentage points.} As time goes by, the gap between the `yes' and `no' groups becomes smaller, implying that some reversion is taking place. However, even after 7 days have elapsed, a substantial gap of approximately 7.3 percentage points can be observed between groups.

\begin{figure}[H]
    \centering
    \hspace{0cm}\includegraphics[width=0.9\textwidth]{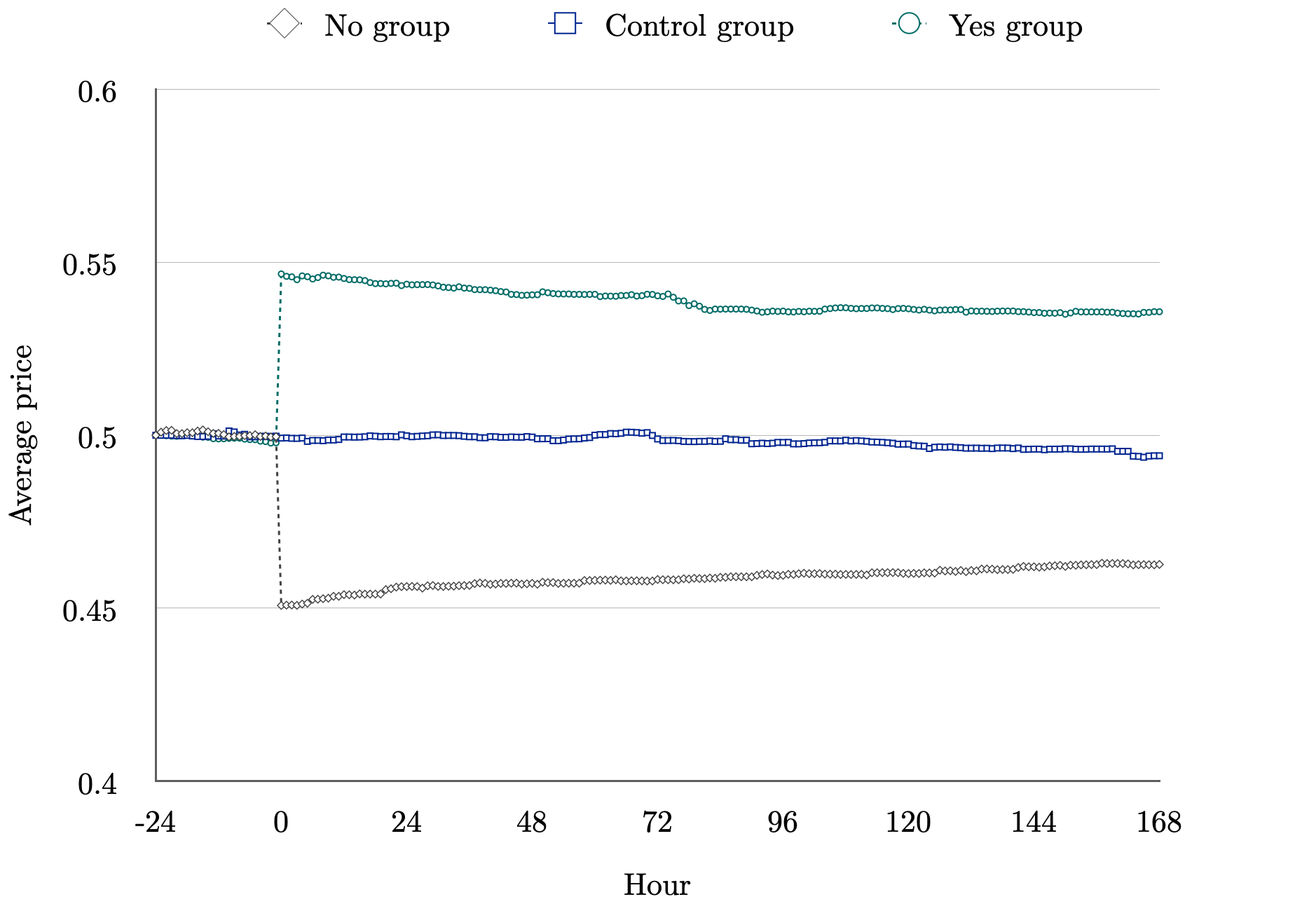}
    \caption{Average prices over time}
    \label{fig:avg_prices}
    \hspace{0cm}\begin{minipage}{0.75\textwidth}
    \vspace{0.3cm}
    \footnotesize \textit{Notes}. This figure shows average prices over time in the three treatment groups. All series are rebased (by adding a constant) so that they start at $0.5$.
\end{minipage}
\end{figure}

To examine this issue more formally, we estimate models of the form
\begin{equation}\label{eq_reg_main}
p_{t, i} = \beta_0 + \beta_1 \mathbbm{1}_{Y, i} +  \beta_2 \mathbbm{1}_{C, i} + \beta_3 p_{-1, i} + u_i
\end{equation}
where $p_{t, i}$ is the price in market $i$ at time $t$, $\mathbbm{1}_{Y, i}$ is a dummy variable that equals $1$ if market $i$ is in `yes' group, $\mathbbm{1}_{C, i}$ is a dummy variable that equals $1$ if market $i$ is in the control group, and $p_{-1, i}$ is the price in the market in the hour before our bet. Thus, the `no' group is the omitted category. As explained in our pre-registration, we include the baseline control $p_{-1, i}$ to increase statistical power. In subsequent robustness checks, we also control for the full suite of baseline variables that are available in our dataset.

Figure \ref{fig:yes_no} presents the estimated $\hat{\beta_1}$ coefficients obtained from estimating this regression for all $t \in \{0, 1, ..., 167\}$. Note that this estimate is simply the difference in average prices between the `yes' and `no' groups, correcting for any baseline imbalance in the $p_{-1, i}$ variable. Again, one can see that the difference starts at around 10 percentage points (immediately after the intervention) and declines over the course of the 1 week period. However, there is a clear difference between the groups of around 7.5 percentage points even after 1 week. As one can see from the 95\% confidence intervals (also depicted in the figure), this difference is statistically distinguishable from zero ($p < 0.01$).

\begin{figure}[h]
    \centering
    \hspace{0cm}\includegraphics[width=0.9\textwidth]{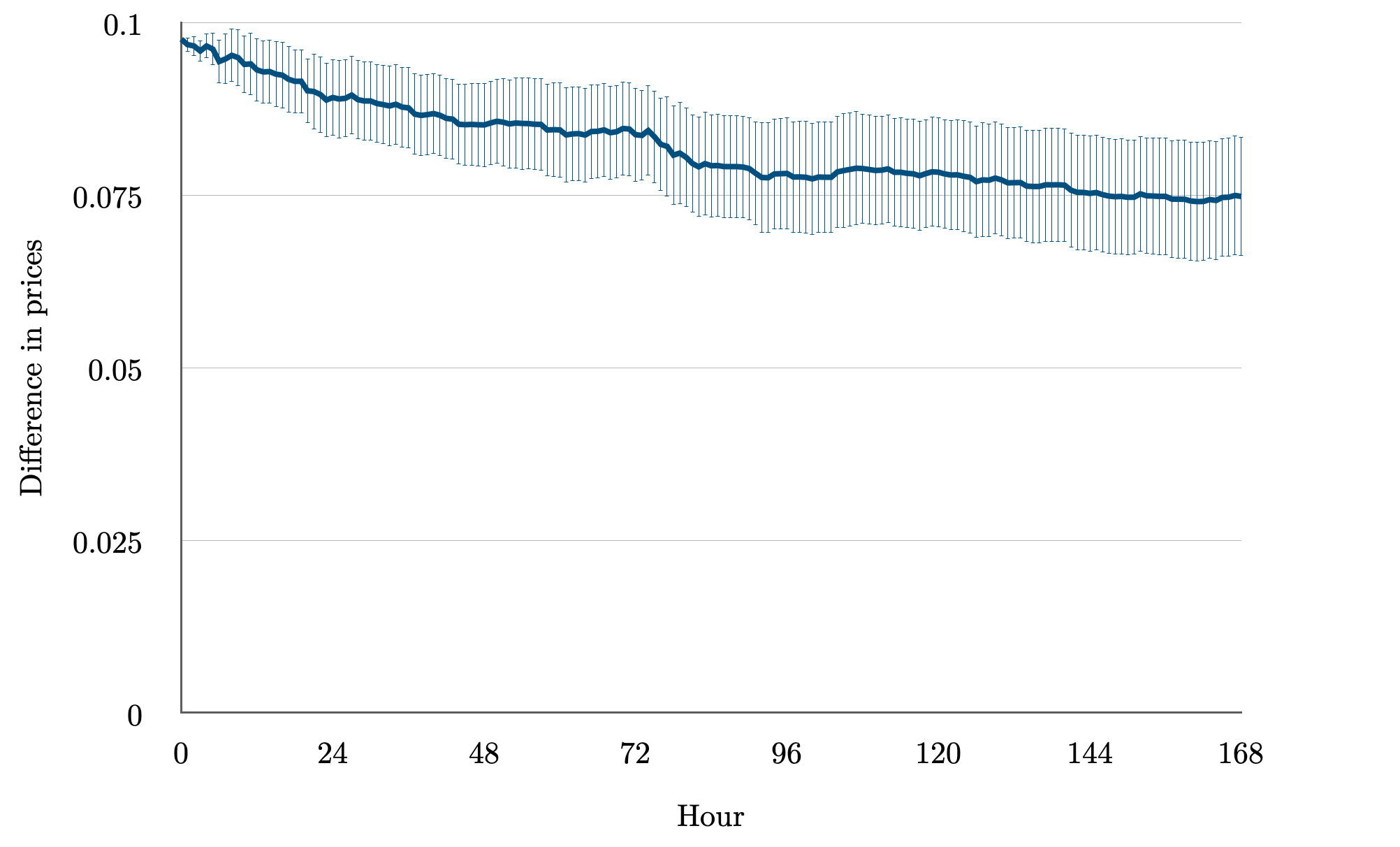}
    \caption{Comparing the `yes' and `no' groups ($\hat{\beta}_1$)}
    \label{fig:yes_no}
    \hspace{0.2cm}\begin{minipage}{0.7\textwidth}
    \vspace{0.15cm}
    \footnotesize \textit{Notes}. This figure plots the $\hat{\beta}_1$ coefficients obtained from estimating \eqref{eq_reg_main} for all $t \in \{0, 1, ..., 167\}$. The 95\% confidence intervals are computed using robust standard errors.
\end{minipage}
\end{figure}

Figure \ref{fig:no_control} presents analogous results for the $\hat{\beta_2}$ coefficients, i.e. the comparison of the `no' group with the control. As can be seen, the results are broadly similar: some reversion is observed, but a sizeable (and statistically significant) difference between the groups can be observed even after 7 days. Note that both of these findings --- some reversion, along with a persistent effect of our manipulations --- are consistent with the model that is outlined in Section \ref{sec_theory}.

Before proceeding, we use the control group to investigate whether `yes' and `no' bets have symmetric effects (as one might expect). To do this, notice that the hypothesis of symmetry is equivalent to the hypothesis that $\beta_1 = 2\beta_2$, which can be tested as a linear restriction on the statistical model. Using an $F$-test, we are unable to reject this restriction at conventional significance levels ($p = 0.124$), which suggests that the observed effects are reasonably symmetric.

\textbf{30 day effects}. Although our pre-registration focused on 7 day effects for reasons of statistical power, one can also conduct similar analyses for the full 30 day horizon by estimating the regression for all $t \in \{0, 1, ..., 719\}$. As shown by Figure \ref{fig:yes_no_30d}, one observes some additional reversion over the longer time horizon. However, even after 30 days, the estimated difference between the `yes' and `no' groups has only fallen to 6.9\%. Remarkably, even though standard errors have unsurprisingly grown due to the accumulation of shocks of time, the effect remains highly significant ($p < 0.01$). It is also worth noticing that the speed of reversion significantly slows over time: we observe 25\% reversion in the first week, but just an additional 6\% reversion over the next three weeks. This reduction in reversion speed is exactly the pattern that is predicted in the model of Section \ref{sec_theory} (see Figure \ref{fig:res_sim}).

\textbf{60 day effects.} Although we had originally planned to only collect data for 30 days after the shocks (as outlined in our analysis plan), the strength of the results after 30 days suggested that very long-run effects might be observable. As a result, we decided to additionally record the price in every market 60 days after the bets. Note that, although our dataset contains hourly data over the course of the first 30 days, we only recorded a single `price snapshot' to measure 60 day effects.

Table \ref{tab:60_day} displays the results. As one would expect, the standard error of $\hat{\beta}_1$ has risen further over time (increasing from .004 after 1 week to .009 after 30 days and .014 after 60 days). Despite this, however, significant effects (with $p < 0.01$) are clearly visible in the data even 60 days after the bets have been made. As can be seen from Column (1), the coefficient $\hat{\beta}_1$ has fallen to .056; that is, prices have now reverted by approximately 44\%. However, this includes markets that had already resolved within 60 days; and one would not expect to be able to manipulate the `60 day price' of these markets (which should generally end up close to 0\% or 100\% regardless of our intervention). Indeed, no significant effects are observed on the `60 day price' of already resolved markets; see Column (2). Once these markets are excluded from the sample, the coefficient on $\hat{\beta}_1$ becomes .059, which suggests that only 41\% of reversion has taken place; see Column (3) for the details. Thus, the effects of our manipulations can be clearly seen in the data even 60 days after they have taken place, especially if one restricts attention to markets that had not resolved by that time.

\begin{table}[h]\caption{60 day effects}\label{tab:60_day}
\vspace{-0.2em}
\begin{threeparttable}
\begin{tabular}{lccc}
\hline
            & (1)          & (2)          & (3)          \\ \hline
Variable    & 60 day effects & 60 day effects & 60 day effects \\
Yes         & 0.0559***    & 0.0855       & 0.0586***    \\
            & {[}0.0135{]} & {[}0.183{]}  & {[}0.0114{]} \\
Control     & 0.0122       & -0.155       & 0.0258**     \\
            & {[}0.0135{]} & {[}0.138{]}  & {[}0.0103{]} \\
$p_{-1, i}$ & 1.020***     & 0.860***     & 1.035***     \\
            & {[}0.0258{]} & {[}0.277{]}  & {[}0.0209{]} \\
Constant    & -0.0474***   & 0.0835       & -0.0596***   \\
            & {[}0.0141{]} & {[}0.135{]}  & {[}0.0110{]} \\ \hline
$n$         & 817          & 46           & 771          \\
$R^2$       & 0.617        & 0.156        & 0.723        \\ \hline
\end{tabular}
\begin{tablenotes}
\footnotesize
\item \hspace{-0.2em}\textit{Notes}. This table shows the results of estimating regression \eqref{eq_reg_main} after 60 days. Column (1) reports results for the full sample; Column (2) for markets that had already resolved; and Column (3) for markets that had not already resolved. Robust standard errors in parentheses (*** $p<0.01$, ** $p<0.05$, * $p<0.1$).
\end{tablenotes}
\end{threeparttable}
\end{table}

\textbf{Robustness.} The previous sections demonstrate that, on average, market prices somewhat revert following manipulative trades. However, even 60 days after the trades have taken place, their impacts can still be clearly seen in the data. Since this result contrasts with some previous findings on prediction market manipulability (e.g., \citealp{rhode2006manipulating}), it is important to assess its robustness. To do this, we focus on the comparison of the `yes' and `no' groups within 30 days of our bets (in line with our pre-registered analysis plan).

First, we re-estimate our main specification while controlling for \textit{all} baseline variables that are available. The results are shown in Columns (2) and (4) of Table \ref{tab:controls}; as a comparison, Columns (1) and (3) show the results from the baseline specification (which just controls for $p_{-1, i}$). As can be seen, adding the baseline controls makes very little difference to the estimates: for instance, the estimated level of reversion changes from 25\% to 24\% after 1 week. In addition, the standard errors of $\hat{\beta_1}$ remain very stable once the controls are added, suggesting that controlling for variables above and beyond the previous price $p_{-1, i}$ adds little statistical power.

\begin{table}[H]
\begin{threeparttable}
\caption{Results with all baseline controls}\label{tab:controls}
\begin{tabular}{lcccc}
\hline
            & (1)           & (2)           & (3)           & (4)           \\ \hline
Variable    & 1 week effects        & 1 week effects        & 30 day effects       & 30 day effects       \\
Yes         & 0.0748***     & 0.0759***     & 0.0686***     & 0.0697***     \\
            & [0.00446] & [0.00431] & [0.00908] & [0.00912] \\
Control     & 0.0313***     & 0.0305***     & 0.0239***     & 0.0235***     \\
            & [0.00484] & [0.00490] & [0.00853] & [0.00854] \\
$p_{-1, i}$ & 1.000***      & 1.003***      & 1.028***      & 1.030***      \\
            & [0.00781] & [0.00781] & [0.0161]  & [0.0155]  \\
Constant    & -0.0371***    & -0.0297***    & -0.0525***    & -0.0644***    \\
            & [0.00547] & [0.0112]  & [0.00910] & [0.0198]  \\ \hline
$n$         & 817           & 817           & 817           & 817           \\
$R^2$       & 0.933         & 0.937         & 0.798         & 0.803         \\\hline
\end{tabular}
\begin{tablenotes}
\footnotesize
\item \hspace{-0.2em}\textit{Notes}. This table shows the estimated $\hat{\beta}_1$ coefficients obtained for $t = 167$ (Columns 1 and 2) and $t = 719$ (Columns 3 and 4). Columns (1) and (3) present the results under the main specification; whereas Columns (2) and (4) control for all baseline variables. Robust standard errors in parentheses (*** $p<0.01$, ** $p<0.05$, * $p<0.1$).
\end{tablenotes}
\end{threeparttable}
\end{table}

Next, we examine whether our results might be biased by spillovers; i.e. by the possibility that betting on one market influences the price of another market in our sample due to either arbitrage or informational effects. As our analysis in Appendix \ref{sec_robustness} makes clear, it is extremely unlikely that spillovers are responsible for our results. First, since we deliberately avoided selecting `highly correlated' markets, it is unclear \textit{a priori} whether our interventions generated any spillovers. Second, even if they did generate spillovers, these should cancel out due to the randomness of the treatment assignment (see Appendix \ref{sec_robustness} for elaboration). Third, and perhaps most importantly, if spillovers are driving the results, then one would expect that dropping the more closely related markets would change the estimates. However, dropping such markets turns out not to have any appreciable effect on our estimates.

\subsection{Heterogeneity}

While the previous section reports the results for the full sample, the model we develop in Section \ref{sec_theory} suggests that the degree of reversion that is observed should systematically depend on market type. First, the model predicts that markets with less scope for learning should be harder to manipulate. To operationalise this notion, we consider whether each market question was `duplicated' on Metaculus since this provides an external probability estimate and thus should lessen the need for traders to base their personal probability estimates on Manifold's market price.\footnote{Whether a market is duplicated on Metaculus is also a reasonable proxy for whether the market can also be found on other prediction market platforms, which in turn should also blunt price learning effects.} Second, we consider each market's level of `activity' as measured by its total trading volume, trading volume in the last 24 hours, and total number of user comments. As discussed earlier, more active markets should revert more quickly since they effectively run on a faster clock speed. Third, we consider each market's number of traders; as discussed in Section \ref{sec_theory}, markets with more traders should also be harder to manipulate.

As outlined in our pre-registration \citep{pre-reg}, we approach this question through median splits. That is, for every continuous moderating variable $v_i$, we construct a variable $x_i$ that equals $1$ if and only if a market's  $v_i$ value is at least the median of $v_i$.\footnote{In the case where the majority of entries are zero, we define the median as the smallest non-zero number and split the sample at this point.} When $v_i$ is binary, we simply set $x_i = v_i$. We then estimate the regression
\begin{align}
p_{167, i} &= \beta_0 + \beta_1 \mathbbm{1}_{Y,i} + \beta_2 \mathbbm{1}_{C,i} + \beta_3 p_{-1, i} + \beta_4 \mathbbm{1}_{Y,i} x_i + \beta_5 \mathbbm{1}_{C,i} x_i + \beta_6 p_{-1, i} x_i + u_i
\end{align}
where $\mathbbm{1}_{Y,i}$ and $\mathbbm{1}_{C,i}$ are defined as before. Notice that the coefficient $\beta_4$ is precisely the difference in treatment effects between markets with $x_i = 0$ (below median) and markets with $x_i = 1$ (above median). By obtaining the standard error of $\hat{\beta_4}$, one can then assess if any differences that are observed are statistically significant.

Table \ref{tab:heterogeneity} reports the 1 week treatment effects conditional on $x_i = 0$ and $x_i = 1$ (Columns 1 and 2); it also displays the differences between treatment effects  $\hat{\beta}_4$ along with the $p$-values corresponding to the test that these differences are zero. Three results are apparent. First, markets that are also present on Metaculus appear to be harder to manipulate: after 1 week, they exhibit 47\% reversion in contrast to the 23\% reversion that is observed for the non-Metaculus markets. While this result is in line with theoretical expectations, it should be noted that, statistically, the difference between the two groups is only significant at the 10\% level. Second, markets with higher levels of `activity' (e.g., higher trading volume) are also harder to manipulate. For example, markets with above median total volume exhibit 31\% reversion over the course of a week; whereas markets with below median total volume exhibit just 19\% reversion. Again, these results are in line with theoretical expectations but not always significant: one observes highly significant results when examining 24 hour volume, but somewhat less significant results when examining total volume and the total number of comments. Third, as one would expect, markets with a larger number of traders are also harder to manipulate: reversion for markets with an above-median number of traders is 33\%, whereas reversion for markets with a below-median number of traders is 17\%.\footnote{While all three results are consistent with our model, they come with the usual caveat concerning causality: since market characteristics are not randomly assigned, one cannot be sure that any observed heterogeneities reflect the causal impact of these characteristics.}

%Fix something weird
\renewcommand{\arraystretch}{1.22}
\begin{table}[H]
\begin{threeparttable}
\caption{Heterogeneity in treatment effects}\label{tab:heterogeneity}
\vspace{-0.5em}
\begin{tabular}{l|cccc}
\multicolumn{1}{l|}{Variable} & \multicolumn{1}{l}{Above median} & \multicolumn{1}{l}{Below median} & \multicolumn{1}{l}{Difference} & \multicolumn{1}{l}{$p$-value} \\ \hline
Metaculus             & 0.053***                            & 0.077*** & 0.024* &  0.096                           \\
24 hour volume        & 0.049***                            & 0.081*** & 0.032** &  0.041                          \\
Total volume          & 0.069***                            & 0.081*** & 0.012 & 0.191                           \\
Comments              & 0.069***                            & 0.084*** & 0.015*  & 0.071 \\
Total traders         & 0.067***                            & 0.083*** & 0.017* &  0.056    \\
\hline
\end{tabular}
\begin{tablenotes}
\footnotesize
\item \hspace{-0.2em}\textit{Notes}. This table shows how results depend on market type. The first two columns report 1 week effects for sub-samples of the data that are above and below the median values of the moderating variable. The next two columns report the difference between these effects along with the $p$-value corresponding to the test that this difference is zero. The asterisks correspond to the hypothesis that the relevant coefficient equals zero (*** $p<0.01$, ** $p<0.05$, * $p<0.1$).
\end{tablenotes}
\end{threeparttable}
\end{table}
%Back to normal...
\renewcommand{\arraystretch}{1.22}

While Table \ref{tab:heterogeneity} reports how the estimates change following median splits, one can also split the data in a more granular way to obtain a more detailed understanding of the heterogeneities at play. Tables \ref{tab:trader_splits} and \ref{tab:volume_splits} conduct this exercise for the two moderating variables with well-defined percentile values: the number of traders and the total volume of trade. Specifically, we split the sample according to whether each of these moderating variables exceeds their 0th percentile (the full sample), their 25th percentile, their 50th percentile (the median split), and their 75th percentile values.\footnote{Note that this exercise cannot be conducted for all variables since they need not have well-defined percentiles: for example, the Metaculus variable is binary.} In line with the predictions of our model (see Table \ref{tab_full_sim}), this produces an apparently monotone pattern: as the number of traders or total volume in the sample rises, the 1 week treatment effect falls. For example, while reversion is 25\% after a week for the full sample, reversion rises to 37\% for the sample of markets that are in the top 25\% in terms of number of traders; note that such markets have 56 traders on average and 349 final trades. However, even once one restricts attention to subgroups with a very high total volume or number of traders (which involves a substantial reduction of the sample), one can still observe clearly significant treatment effects ($p < 0.01$).

In summary, the evidence presented here suggests that the manipulability of markets varies across the sample in just the way that our model predicts. In all sub-samples that we consider, we observe the same qualitative pattern: some reversion over time, but also some scope for manipulation. However, the degree of reversion varies in the expected ways. In particular, markets with an external source of probability estimates (namely, Metaculus), markets with higher `activity' (measured by market volume and the number of comments) and markets with a larger number of traders appear to be harder to persistently manipulate.

\section{Sweepcash markets}\label{sec_sweepcash}

Before concluding, we briefly discuss the results of a follow-up experiment that was made possible by an unexpected change to the Manifold platform. On 18 September 2024, Manifold announced that it would introduce `Sweepcash' markets to run in parallel with the `Mana' markets that it was already hosting on its website. Such markets would run on a Manifold created currency (`Sweepcash') that was convertible at an approximately 1-1 rate with US dollars.\footnote{At the time of our experiment, it was possible to buy 1 unit Sweepcash for 1 US dollar (although slightly different prices were also available depending on the purchase quantity). However, Manifold levied a 5\% fee on all Sweepcash withdrawals.} Interestingly, Sweepcash markets were to work almost exactly like Mana markets in essentially all other respects --- for example, they were to run on the same interface and were underpinned by the same pricing algorithm that we study theoretically in Section \ref{sec_theory}. Thus, the introduction of Sweepcash markets provided an opportunity to learn whether the manipulability of prediction markets might change once the `charitable giving' incentive was replaced by more standard financial incentives.

\textbf{Intervention.} Unfortunately, the number of Sweepcash markets offered by Manifold was rather more limited than the number of Mana markets hosted on the platform. For this reason, we decided to drop the control group from our follow-up; instead, all markets were either allocated to a `yes' or `no' group. Since our working paper was already circulating at this point, we were concerned that a string of bets that all moved prices by exactly 5 percentage points might create suspicion. For this reason, we decided to change our shocks to 4 percentage points. Thus, markets in the `yes' group had their prices shocked upwards by 4 percentage points and markets in the `no' group had their prices shocked downwards by the same amount.

\textbf{Exclusion criteria.} Similarly to the main experiment, we excluded some markets from our sample. First, we excluded markets that resolved after the end of 2026 or within 14 days of our trade. Second, we again excluded markets that were closely related to other markets in our sample. Third, we excluded any market that would cost more than \$500 to move in either direction by 4 percentage points. In the event, this constraint was not binding; our largest trade (for a particular market) was \$175.\footnote{On average, we spent \$15.72 per trade; and thus invested \$1,729.20 in the Sweepcash markets in total.}

\textbf{Power.} To choose a sample size, we calibrated our assumed statistical model using data from our previous experiment on Manifold; see our pre-registration \citep{pre-reg2} for the details. As it turned out, this calibration was reasonably accurate: for example, we assumed an error variance of $0.074$, which closely matches up with the (ex-post) estimate from our residuals of $0.076$. Based on the simulations reported in \cite{pre-reg2}, we estimated that we would need 110 markets in order to have an 80\% probability of detecting effects within a one week period. While one would ideally have more than 80\% power, a larger sample was precluded by the limited number of Sweepcash markets on Manifold.\footnote{When we started our experiment, there were only 69 Sweepcash markets on Manifold that met our inclusion criteria. We suspected (correctly, as it turned out) that Manifold would be slow to create additional Sweepcash markets.}

\textbf{Data.} As before, we recorded extensive data on the markets, including the number of comments, total volume and number of traders at the point at which we made our trades (see \citealp{pre-reg2} for the list of price variables). Most importantly, we also recorded one week's worth of price data for each market, this time on a daily basis. Given our relatively small sample size---and that power to detect effects generally falls as the time horizon grows---we decided not to record price data beyond this one week window.

\textbf{Timelines.} We pre-registered our experiment on 3 December and made our trades from December 2024 to February 2025. We had planned to finish our data recording on 13 February 2025. Unfortunately, on that same day, Manifold announced that it was to discontinue Sweepcash markets and suspended trading on most of the Sweepcash markets in operation. For this reason, we were not able to collect seven days worth of price data for all of the markets in our sample. However, we do have six days worth of price data for the full sample; and also have a subset of markets with seven days worth of price data.

\textbf{Results.} Table \ref{tab:desc_stat2} shows some descriptive statistics on the markets in our sample. As before, the first four variables were measured just before our trades; whereas the final two variables show the total number of traders and trades by the end of the markets' life-cyles (for those markets that had `naturally' resolved before the suspension of Sweepcash markets). As can be seen, the Sweepcash markets were somewhat more active on average than the Mana markets in our main experiment. For example, the average market had 158.7 trades by the end of its lifespan, compared to a corresponding figure of 136.0 for the Mana markets.

\begin{table}[H]
\begin{threeparttable}
\caption{Descriptive statistics}\label{tab:desc_stat2}
\vspace{-0.5em}
\begin{tabular}{lcccc}
\hline
Variable       & \hspace{1.5em}Mean\hspace{1.5em}     & Std. dev. & Minimum & Maximum   \\ \hline
24 hour volume &  21.2  & 46.3 & 0 & 277 \\
Total volume   & 176 & 322 & 2 & 2,523 \\
Comments       & 6.24    & 11.12              & 0   & 67    \\
Traders        & 13.8   & 15.9             & 1  & 95   \\ \hline
Final traders  & 49.4 & 42.7        & 18  &  243   \\
Final trades   & 158.7  &  123.5            & 29  &  494  \\ \hline
\end{tabular}
\begin{tablenotes}
\footnotesize
\item \hspace{-0.2em}\textit{Notes}. This table shows descriptive statistics for the Sweepcash markets in our sample. The first four variables show statistics for the full sample ($n = 110$). The last two variables show statistics for the 35 markets that had resolved before the suspension of Sweepcash markets.
\end{tablenotes}
\end{threeparttable}
\end{table}

Similarly to before, we assess the effect of our manipulation attempts by estimating models of the form
\begin{equation}\label{eq_reg_main2}
p_{t, i} = \beta_0 + \beta_1 \mathbbm{1}_{Y, i} + \beta_3 p_{-1, i} + u_i
\end{equation}
where $p_{t, i}$ is the price in market $i$ on day $t$, $\mathbbm{1}_{Y, i}$ tracks whether a market was placed in the `yes' group, and $p_{-1, i}$ is the price in the market $i$ just before our trade. Note that $\beta_1$ again represents the difference in average prices between the `yes' and `no' groups, adjusting for any baseline differences in the $p_{-1, i}$ variable.

As Figure \ref{fig:yes_no_followup} reveals, two results emerge from this exercise. First, we again see clear evidence of manipulability: on every day for which data is recorded, prices in the `yes' group are higher than prices in the `no' group. Second, we again see clear evidence of reversion: as time goes on, the gap in average prices falls substantially from 8 percentage points. It is perhaps worth noting that the estimated speed of reversion appears to slow down over time, in agreement with our theoretical model. However, the size of the confidence intervals makes it difficult to say anything definite about curvature.

\begin{figure}[t!]
    \centering
    \hspace{0cm}\includegraphics[width=0.95\textwidth]{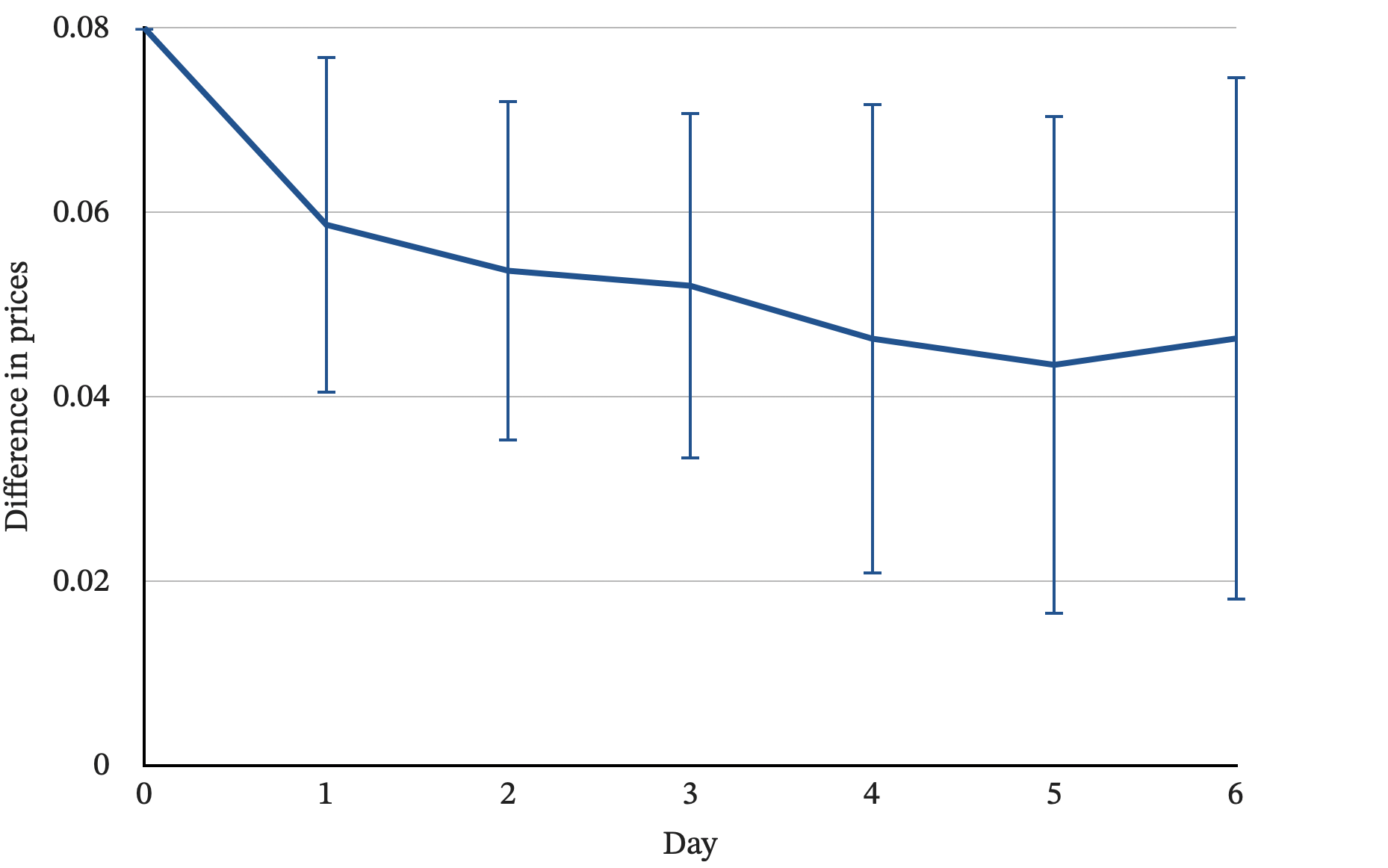}
    \caption{Results for the Sweepcash markets ($\hat{\beta}_1$ estimates)}
    \label{fig:yes_no_followup}
    \hspace{0cm}\begin{minipage}{0.85\textwidth}
    \vspace{0.15cm}
    \footnotesize \textit{Notes}. This figure plots the $\hat{\beta}_1$ coefficients obtained from estimating Equation \eqref{eq_reg_main2} for all $t \in \{0, 1, ..., 6\}$. The 95\% confidence intervals are computed using robust standard errors.
\end{minipage}
\end{figure}

To check the robustness of our results, we can also estimate analogous $\hat{\beta}_1$ coefficients while controlling for all baseline variables. As one can see from Table~\ref{tab:reg_sweepcash}, this does not have any real impact on the estimates; as before, adding further controls (above and beyond $p_{-1, i}$) does not add additional power. We can also estimate effects for an extra day using the subset of markets for which such data is available. As can be seen from Table~\ref{tab:reg_sweepcash}, this specification again shows clear evidence of reversion along with clear evidence of manipulability.

Interestingly, the degree of reversion observed in the Sweepcash markets appears to be faster than the degree observed in the full sample of the main experiment. However, the size of our standard errors (given the relatively small sample in the follow-up) makes it hard to establish this with confidence. One explanation for this difference, assuming that it is real, would be that the Sweepcash markets tend to be more active than the average market from the main sample. Indeed, if one compares markets with positive 24 hour volume from both samples, one obtains reasonably similar reversion estimates (51\% reversion for the Mana markets and 43\% reversion for the Sweepcash markets).

One interesting question is how our $\hat{\beta}_1$ estimates might change if our design were extended to some of the very high-value prediction markets that are currently in operation. Our model would predict that a shock of a particular size (e.g. 4 percentage points) would be more persistent in such markets since it is more expensive and thus harder to `undo' by risk-averse or budget-constrained traders. In addition, one might expect that a very expensive trade in a high-value market conveys a stronger informational signal, thereby increasing the extent to which it updates trader beliefs and thus persists. While it would be interesting to test these conjectures in future work, doing so would be both expensive and logistically challenging (not least because there are few very high-value prediction markets in operation).

Taken together, the results from this section suggest that the results from the main experiment extend robustly to Sweepcash markets. As before, we obtain clear evidence that the markets are manipulable---as indicated by the gap in average prices between the `yes' and `no' groups---along with clear evidence of reversion. Given these results, it does not appear that replacing Mana with Sweepcash (while holding most other features of the platform constant) has a substantial impact on prediction market manipulability.

\section{Conclusion}\label{sec_conclusion}

In their review of the existing evidence, \cite{wolfers2004prediction} state that manipulation attempts do not have `much of a discernible effect on prices, except during a short transition phase'. Our large-scale field experiments challenge this conclusion: for example, in our main experiment, we can detect the effects of our manipulations even 60 days after they were made. However, as predicted by our model, we also find substantial reversion and important heterogeneities in the expected directions.

In our view, our results somewhat understate the degree to which prediction markets are susceptible to manipulation. To keep our design simple and easily interpretable, we made just one trade within each market and shocked the markets in our sample by the exact same amount (which, in theory, might seem suspicious to other traders). In practice, however, a manipulator might be able to move prices even more effectively by making repeated bets on a single market and varying the size of their shocks.

%Our findings somewhat confirm the manipulability concerns raised by prediction markets' critics. However, they do \textit{not} imply that prediction markets are unhelpful or that they should be regulated further: even if manipulable, their prices can still be somewhat informative \citep{hanson2004foul}. Moreover, although non-causal, our heterogeneity results may suggest that making prediction markets more `active' (by encouraging higher volume, more traders, etc.) can make them more robust to manipulation attempts.

Our experiment also opens the door to a lot of future work. First, although we study `manipulation through trade', it may also be interesting to study `manipulation through buzz': by leaving appropriate comments, it may also be possible to systematically shift the market price.\footnote{We thank Koleman Strumpf for this suggestion.} Second, it would be interesting to study optimal manipulation. While very small trades cannot have a large impact on the market price, very large trades also appear unlikely to have persistent effects since they are unlikely to appear `credible' --- for this reason, one might expect a `U-shaped' relationship between the long-run effect on prices and size of the initial shock. Although our design does not address this question since it holds the size of the initial shock fixed, it would be straightforward (if logistically demanding) to adapt our design so that it considers a range of initial shock sizes.

\clearpage

\setlength{\bibhang}{0pt}
\bibliographystyle{apalike}
\bibliography{bibliography.bib}

\newpage

\setcounter{table}{0}
\renewcommand{\thetable}{A\arabic{table}}

\setcounter{figure}{0}
\renewcommand{\thefigure}{A\arabic{figure}}

\AddToHook{env/appendices/begin}{%
  \titleformat{\section}
  {\normalfont \scshape}{\sectionname~\thesection.}{0.5em}{}
}

\begin{appendices}

\section{Proofs}\label{sec_proofs}

\begin{proof}[Proof of Lemma \ref{lemma1}]
We begin by considering the marginal cost function. As shown in the main text, the total cost function is given by
\begin{equation}
C(q) =  \frac{\sqrt{(n  - q + y)^2 + 4 n q} + q - n - y}{2}
\end{equation}
This implies that marginal costs are given by
\begin{equation}\label{mc}
MC(q) = \frac{1}{2} \left(\frac{n + q - y}{ \sqrt{(n-q+y)^2+4 n q}}+1\right)
\end{equation}
The three properties follow from this formula:

(i) To show that $MC(0) = n/(n + y)$, one simply substitutes  $q = 0$ into (\ref{mc}) and simplifies the resulting expression.

(ii) To show that $MC'(q) > 0$ for all $q \geq 0$, one computes the derivative
\begin{equation}
MC'(q) = \frac{2 n y}{\left((n-q+y)^2+4 n q\right)^{3/2}}
\end{equation}
which is positive under the restrictions that $n > 0$, $y > 0$ and $q \geq 0$.

(iii) To show that $\displaystyle{\lim_{q \rightarrow \infty} MC(q) = 1}$, one notes that
\begin{equation}
\begin{split}
\lim_{q \rightarrow \infty} MC(q) &= \frac{1}{2} \lim_{q \rightarrow \infty} \frac{n + q - y}{ \sqrt{(n-q+y)^2+4 n q}} + \frac{1}{2} \\
&= \frac{1}{2} \lim_{q \rightarrow \infty} \sqrt{\frac{(n + q - y)^2}{ (n-q+y)^2+4 n q}} + \frac{1}{2} \\
&= \frac{1}{2} \sqrt{\frac{1}{1}} + \frac{1}{2} \\
&= 1
\end{split}
\end{equation}
where the penultimate equality follows after expanding out the terms and dividing all terms by the leading power $q^2$. This establishes the final property of the marginal cost function.

We now show that these properties transfer to the average cost function:

(i) To show that $\displaystyle{\lim_{q \rightarrow 0^+}} AC(q) = n/(n + y)$, note that
\begin{equation}
\lim_{q \rightarrow 0^+} AC(q) = \lim_{q \rightarrow 0^+} \frac{C(q)}{q}
= \lim_{q \rightarrow 0^+} \frac{C(0 + q) - C(0)}{q}
= MC(0)
= \frac{n}{n + y}
\end{equation}
where the second equality uses the fact that $C(0) = 0$.

(ii) To show that $AC'(q) > 0$ for all $q > 0$, we explicitly compute the derivative
\begin{equation}
\frac{\partial AC}{\partial q} = \frac{1}{2q^2} \left( -\sqrt{(n - q + y)^2 + 4 n q} + \frac{q (n + q - y)}{\sqrt{(n - q + y)^2 + 4 n q}} + n + y \right)
\end{equation}
Multiplying by $2q^2 \sqrt{(n - q + y)^2 + 4 n q} > 0$, this has the same sign as
\begin{align}
-(n - q + y)^2 - 4nq + q (n + q - y) + (n + y)\sqrt{(n - q + y)^2 + 4 n q}
\end{align}
We want to show that this expression is positive, or
\begin{align}(n + y)\sqrt{(n - q + y)^2 + 4 n q} > -[q (n + q - y) -(n - q + y)^2 - 4 n q]
\end{align}
After squaring both sides and simplifying, the condition becomes $n q^2 y > 0$, which holds since $n, q, y > 0$. Thus, $AC'(q) > 0$ for all $q > 0$ as claimed.

(iii) To show that $\displaystyle{\lim_{q \rightarrow \infty}} AC(q)  = 1$, notice that, as $q \to \infty$, both $C(q) \to \infty$ and $q \to \infty$. Since both are differentiable, we can evaluate the limit of $AC(q)$ using L'Hôpital's rule:
\begin{equation}
\lim_{q \rightarrow \infty} AC(q) = \lim_{q \rightarrow \infty} \frac{C(q)}{q} = \lim_{q \rightarrow \infty} \frac{C'(q)}{1} = \lim_{q \rightarrow \infty} MC(q) = 1
\end{equation}
This establishes the final property of the average cost function.\end{proof}

\begin{proof}[Proof of Lemma \ref{lemma2}]
We consider the cases separately:

\textbf{Case 1}: $\pi_i > p$. First, suppose that the agent purchases a positive quantity $q_n > 0$ of no shares. Then the expected value of their holdings $W$ would be:
\begin{equation}
\begin{split}
\mathbb{E}[W] &= \pi_i (w - q_n AC(q_n)) + (1 - \pi_i)(w - q_n AC(q_n) + q_n) \\
&< \pi_i \left(w - q_n \frac{y}{n + y} \right) + (1 - \pi_i)\left(w - q_n \frac{y}{n + y} + q_n \right) \\
&= w + q_n \left(\frac{n}{n + y} - \pi_i \right) \\
&< w
\end{split}
\end{equation}
Here, the first inequality holds since $AC(q_n) > y/(n + y)$ for any $q_n > 0$ (by the no shares analogue of Lemma \ref{lemma1}); the next equality holds via algebraic manipulation; and the final inequality holds since $\pi_i > p \equiv n/(n + y)$. Thus, holding no shares reduces the expected value of the agent's holdings. Since it also exposes the agent to risk, and the agent is risk-averse, it therefore cannot be optimal.

Next, suppose that the agent purchases a quantity $q_y \geq 0$ of yes shares (possibly, $q_y = 0$). Their expected utility is then
\begin{align}
\mathbb{E}[u(W)] &= \pi_i u(w + q_y - C(q_y)) + (1 - \pi_i)u(w - C(q_y))
\end{align}
Using the chain rule, this has derivative
\begin{align}\label{derivative}
\frac{\partial \mathbb{E}[u(W)]}{\partial q_y} &= \pi_i u'(w + q_y - C(q_y))(1 - C'(q_y))
- (1 - \pi_i)u'(w - C(q_y))C'(q_y)
\end{align}
Recalling Lemma \ref{lemma1}, we see that
\begin{align}
\left. \frac{\partial \mathbb{E}[u(W)]}{\partial q_y} \right|_{q_y=0} &= \pi_i u'(w)\left(1 - \frac{n}{n + y}\right) - (1 - \pi_i)u'(w) \left(\frac{n}{n + y} \right) \\
&= u'(w) \left( \pi_i - \frac{n}{n + y} \right)
\end{align}
Since $u'(w) > 0$ and $\pi_i > n/(n + y)$, this expression is positive and so $q_y = 0$ cannot be optimal: there must exist some $\epsilon > 0$ such that $q_y = \epsilon$ generates higher expected utility.

The previous arguments show that neither $q_n > 0$ nor $q_y = 0$ are optimal choices. However, since the trader's choice set is compact (recall the budget constraint) and their objective function is continuous (recall that $u$ is continuous and that the cost function is polynomial), an optimal choice must exist. From this, it follows that some $q_y > 0$ must be optimal.

\textbf{Case 2 $\pi_i = p$}. Using the same reasoning as before, we see that, if the agent purchases no shares, then the expected value of their holdings will be
\begin{equation}
\mathbb{E}[W] < w + q_n \left(\frac{n}{n + y} - \pi_i \right) = w
\end{equation}
Thus, buying no shares reduces the expected value of their holdings. Since it exposes them to risk, and they are risk-averse, this cannot be optimal. Analogously, if they purchase yes shares, the expected value of their holdings becomes
\begin{equation}
\mathbb{E}[W] < w + q_y \left(\pi_i - \frac{n}{n + y}\right) = w
\end{equation}
Thus, purchasing yes shares also reduces the expected value of their holdings and thus cannot be optimal. However, an optimal choice does exist (see above). From this, it follows that the optimal choice is $q_n = q_y = 0$.

\textbf{Case 3 $\pi_i < p$}. This case can be handled in a similar way to the case of $\pi_i > p$. That is, one first shows that purchasing yes shares cannot be optimal since it reduces the expected value of the agent's holdings. By examining the derivative of expected utility when $q_n  = 0$, one then shows that $q_n = 0$ cannot be optimal either.\end{proof}

\begin{proof}[Proof of Lemma \ref{lemma3}] First, we use Lemma \ref{lemma2} to pin down the types of share that each type of trader must buy. Notice that:
\begin{itemize}
    \item If $\pi_i \geq p + \Delta$, then $\pi_i \geq p$ (since $\Delta > 0$). Thus, Lemma \ref{lemma2} implies that such types must hold a weakly positive quantity of yes shares both before and after the price change.
    \item If $\pi_i \in (p, p + \Delta)$, Lemma \ref{lemma2} implies that such types would buy yes shares before the price change but buy no shares after the price change.
    \item If $\pi_i \leq p$, then $\pi_i \leq p + \Delta$. Thus, Lemma \ref{lemma2} implies that such types must hold a weakly positive quantity of no shares both before and after the price change.
\end{itemize}

Next, we need to show that types with $\pi_i \geq p + \Delta$ will reduce their holdings of yes shares $q$. If $\pi_i = p + \Delta$, the statement is trivial: by Lemma \ref{lemma2}, such types will choose $q > 0$ before the price change but $q = 0$ after the price change. Next, suppose that $\pi_i = 1$. For such types, expected utility reduces to $u(w + q - C(q))$, which can be transformed to $w + q - C(q)$. Differentiating with respect to $q$, this becomes $1 - C'(q) > 0$ since $C'(q)< 1$ for any finite $q$ (by Lemma \ref{lemma1}). Thus, such types will purchase the maximum feasible number of yes shares, i.e. $\bar{q}$ such that $C(\bar{q}) = w$. Since $C$ is strictly increasing in $q$ and so invertible, $\bar{q} = C^{-1}(w)$. Moreover, one can verify (see below) that increasing the price --- which is equivalent to increasing $n$ and decreasing $y$ while holding the product $ny$ constant --- increases $C(q)$ for any $q \geq 0$, thus decreasing $C^{-1}(q)$. Hence, increasing the price decreases $\bar{q} = C^{-1}(w)$; that is, it induces the trader to reduce their holdings of yes shares (as claimed).

Having shown that the statement holds when $\pi_i = p + \Delta$ or $\pi_i = 1$, we now consider traders whose beliefs satisfy $\pi_i \in~(p + \Delta, 1)$. As shown by Lemma \ref{lemma2}, $q = 0$ cannot be optimal for such types. Meanwhile, since $\lim_{w_s \to 0^+} u'(w_s) = \infty$, it cannot be optimal to choose $q = \bar{q}$, thereby obtaining zero wealth in one of the states. Hence, $q \in (0, \bar{q})$. Since expected utility is differentiable, this means that any optimal $q$ must satisfy the first order condition
\begin{equation}\label{foc}
\pi_i u'(w + q - C(q))(1 - C'(q))
- (1 - \pi_i)u'(w - C(q))C'(q) = 0
\end{equation}
Moreover,
\begin{equation}
\begin{split}
\frac{\partial^2 \mathbb{E}[u(W)]}{\partial q^2} &= \pi_i \left[u''(w + q - C(q))(1 - C'(q))^2 -  C''(q)u'(w + q - C(q))\right] \\
&- (1 - \pi_i) \left[-C'(q)^2u''(w - C(q)) + u'(w - C(q))C''(q)\right]
\end{split}
\end{equation}
which is negative under the restrictions that $u' > 0$, $u'' < 0$, $C' > 0$ and $C'' > 0$. Thus, the first order condition (\ref{foc}) uniquely characterises the global maximum.

Rewriting (\ref{foc}), we obtain
\begin{equation}\label{eq_foc'}
\frac{\pi_i u'(w + q - C(q))}{(1 - \pi_i)u'(w - C(q))} = \frac{C'(q)}{1 - C'(q)}
\end{equation}
Notice that, as the price rises --- which is equivalent to increasing $n$ and decreasing $y$ while holding the product $k = ny$ constant --- both $C(q)$ and $C'(q)$ may change. To make this dependence explicit, we can write the first order condition as $f(q, y, n) = g(q, y, n)$, where $f$ and $g$ are respectively defined as the left hand side and right hand side of~\eqref{eq_foc'}. Since $y$ is a function of $n$ (specifically, $y = k/n$), this can be simplified to $f(q, n) = g(q, n)$. To see how the optimal $q$ varies with the price (i.e. $n$), one can apply the implicit function theorem
\begin{equation}\label{ift}
 \frac{d q}{d n} = \frac{\frac{\partial g(q, n)}{\partial n}  - \frac{\partial f(q, n)}{\partial n} }{\frac{\partial f(q, n)}{\partial q} - \frac{\partial g(q, n)}{\partial q}}
\end{equation}
To compute the sign of (\ref{ift}), we now compute the signs of every term:

$\diamond$ Observe that
\begin{equation}\label{eq_dg/dn}
\frac{\partial g(q, n)}{\partial n} = \frac{\partial g(q, n)}{\partial C'(q)} \frac{\partial C'(q)}{\partial n} = \frac{1}{(1 - C'(q))^2} \frac{\partial C'(q)}{\partial n}
\end{equation}
Moreover, after substituting $y = k/n$ into the marginal cost function, one finds that
\begin{equation}\label{eq_dc/dn}
\frac{\partial C'(q)}{\partial n} = \frac{2 k n  \left(k+n^2\right)}{\left(k^2+2 k n (n-q)+n^2 (n+q)^2\right)^{3/2}} > 0
\end{equation}
From \eqref{eq_dg/dn} and \eqref{eq_dc/dn}, it follows that $\frac{\partial g(q, n)}{\partial n} > 0$.

$\diamond$ Next, observe that
\begin{equation}\label{df_dn}
\frac{\partial f(q, n)}{\partial n} = \frac{\partial f(q, n)}{\partial C(q)}  \frac{\partial C(q)}{\partial n}
\end{equation}
Since $C(q) = \int_0^q C'(q) dq$ and $\frac{\partial C'(q)}{\partial n} > 0$, $\frac{\partial C(q)}{\partial n} > 0$. Moreover,
\begin{equation}
\begin{split}
\frac{\partial f(q, n)}{\partial C(q)} &= \frac{\pi_i}{(1 - \pi_i)}\frac{\partial }{\partial C(q)}  \frac{ u'(w + q - C(q))}{u'(w - C(q))} \\
&= \frac{\pi_i}{(1 - \pi_i)} \frac{u''(w - C(q))u'(w + q - C(q)) - u''(w + q - C(q))u'(w - C(q))}{u'(w - C(q))^2}
\end{split}
\end{equation}
which has the same sign as
\begin{equation}
u''(w - C(q))u'(w + q - C(q)) - u''(w + q - C(q))u'(w - C(q))
\end{equation}
We claim that this expression is negative, or
\begin{equation}
-\frac{u''(w - C(q))}{u'(w - C(q))} > -\frac{u''(w + q - C(q))}{u'(w + q - C(q))}
\end{equation}
which holds under the assumption of decreasing absolute risk aversion. Hence, $\frac{\partial f(q, n)}{\partial C(q)} < 0$. From this and (\ref{df_dn}), it follows that $\frac{\partial f(q, n)}{\partial n} < 0$.

$\diamond$ Next, observe that
\begin{equation}
    \frac{\partial f(q, n)}{\partial q} = \frac{\pi_i}{1-\pi_i}\frac{\partial}{\partial q} \frac{u'(w + q - C(q))}{u'(w - C(q))}
\end{equation}
Since $\pi_i/(1 - \pi_i) > 0$, this has the same sign as
\begin{equation}
\frac{u''(w + q - C(q))u'(w - C(q))(1-C'(q)) + u''(w - C(q))u'(w + q - C(q))C'(q)}{u'(w - C(q))^2}
\end{equation}
which is negative given that $u' > 0$, $u'' < 0$ and $C'(q) \in (0, 1)$. Hence, $\frac{\partial f(q, n)}{\partial q} < 0$.

$\diamond$ Finally, observe that
\begin{equation}
\frac{\partial g(q, n)}{\partial q} = \frac{\partial g(q, n)}{\partial C(q)} C'(q) = \frac{C'(q)}{(1 - C'(q))^2} > 0
\end{equation}

In summary, $\frac{\partial g(q, n)}{\partial n} > 0$, $\frac{\partial f(q, n)}{\partial n} < 0$, $\frac{\partial f(q, n)}{\partial q} < 0$ and $\frac{\partial g(q, n)}{\partial q} > 0$. It follows that
\begin{equation}
 \frac{d q}{d n} = \frac{\frac{\partial g(q, n)}{\partial n}  - \frac{\partial f(q, n)}{\partial n} }{\frac{\partial f(q, n)}{\partial q} - \frac{\partial g(q, n)}{\partial q}} < 0
\end{equation}
Thus, any trader with $\pi_i \geq p + \Delta$ will decrease their holdings of yes shares. By an exactly symmetric argument (which we omit for brevity), one can also show that any trader with $\pi_i \leq p$ will increase their holdings of no shares.\end{proof}

\newpage

\section{Tables and figures}\label{sec_tables}

\begin{figure}[h!]
    \centering
    \caption{Costs under the constant product rule ($n = y = 10$)\vspace{-1.5em}}
    \hspace{-2em}\includegraphics[width=14cm, keepaspectratio]{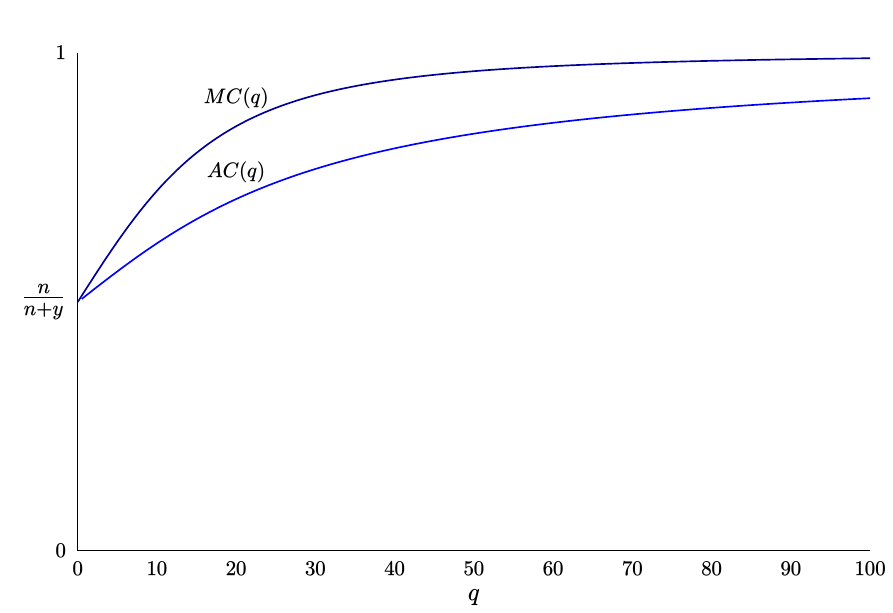}
    \label{fig:mc_ac}
    \begin{minipage}{12.8cm}
    \vspace{0.3cm}
    \footnotesize \textit{Notes}. This figure displays marginal and average costs under the constant product rule with initial reserves $(y, n) = (10, 10)$.
\end{minipage}
\end{figure}

\begin{table}[H]
\begin{threeparttable}
\caption{Reversion coefficients under different market conditions}\label{tab_full_sim}
\begin{tabular}{l|cccccc}
\hline
$\lambda$ & 0      & 0.2    & 0.4    & 0.6    & 0.8    & 1      \\ \hline
SR reversion         & \hspace{0.5em}14.5\%\hspace{0.5em} & \hspace{0.5em}13.3\%\hspace{0.5em} & \hspace{0.5em}11.2\%\hspace{0.5em} & \hspace{0.5em}8.3\%\hspace{0.5em}  & \hspace{0.5em}4.5\%\hspace{0.5em}  & \hspace{0.5em}0.0\%\hspace{0.5em}  \\
LR reversion           & 39.8\% & 37.3\% & 33.0\% & 26.2\% & 15.9\% & 0.0\%  \\ \hline
$\alpha$  & 0      & 0.2    & 0.4    & 0.6    & 0.8    & 1      \\ \hline
SR reversion           & 22.7\% & 22.4\% & 21.5\% & 19.9\% & 17.6\% & 14.5\% \\
LR reversion           & 52.3\% & 51.8\% & 50.6\% & 48.4\% & 44.9\% & 39.8\% \\ \hline
$m$       & 10     & 20     & 30     & 40     & 50     & 60     \\ \hline
SR reversion            & 14.5\% & 16.0\% & 16.7\% & 16.9\% & 18.5\% & 18.6\% \\
LR reversion            & 39.8\% & 57.1\% & 67.1\% & 74.4\% & 79.3\% & 82.0\% \\
\hline
\end{tabular}
\begin{tablenotes}
\footnotesize
\item \hspace{-0em}\textit{Notes}. This table shows the short-run (SR) and long-run (LR) reversion coefficients under different market conditions. SR reversion is defined using $(\bar{p}_3 - \bar{p}_0)/0.05$ and LR reversion using $(\bar{p}_{100} - \bar{p}_0)/0.05$, where $\bar{p}_t$ is the average price across the 10,000 simulations at time $t$ and $t = 0$ is the time of the manipulation. The coefficients $\lambda$, $\alpha$ and $m$ respectively denote the learning rate, the belief parameter and the number of traders.
\end{tablenotes}
\end{threeparttable}
\end{table}

\begin{figure}[H]
    \centering
    \caption{An example of a market on Manifold.}
    \vspace{-0.5em}
    \hspace{0cm}\includegraphics[width=0.9\textwidth]{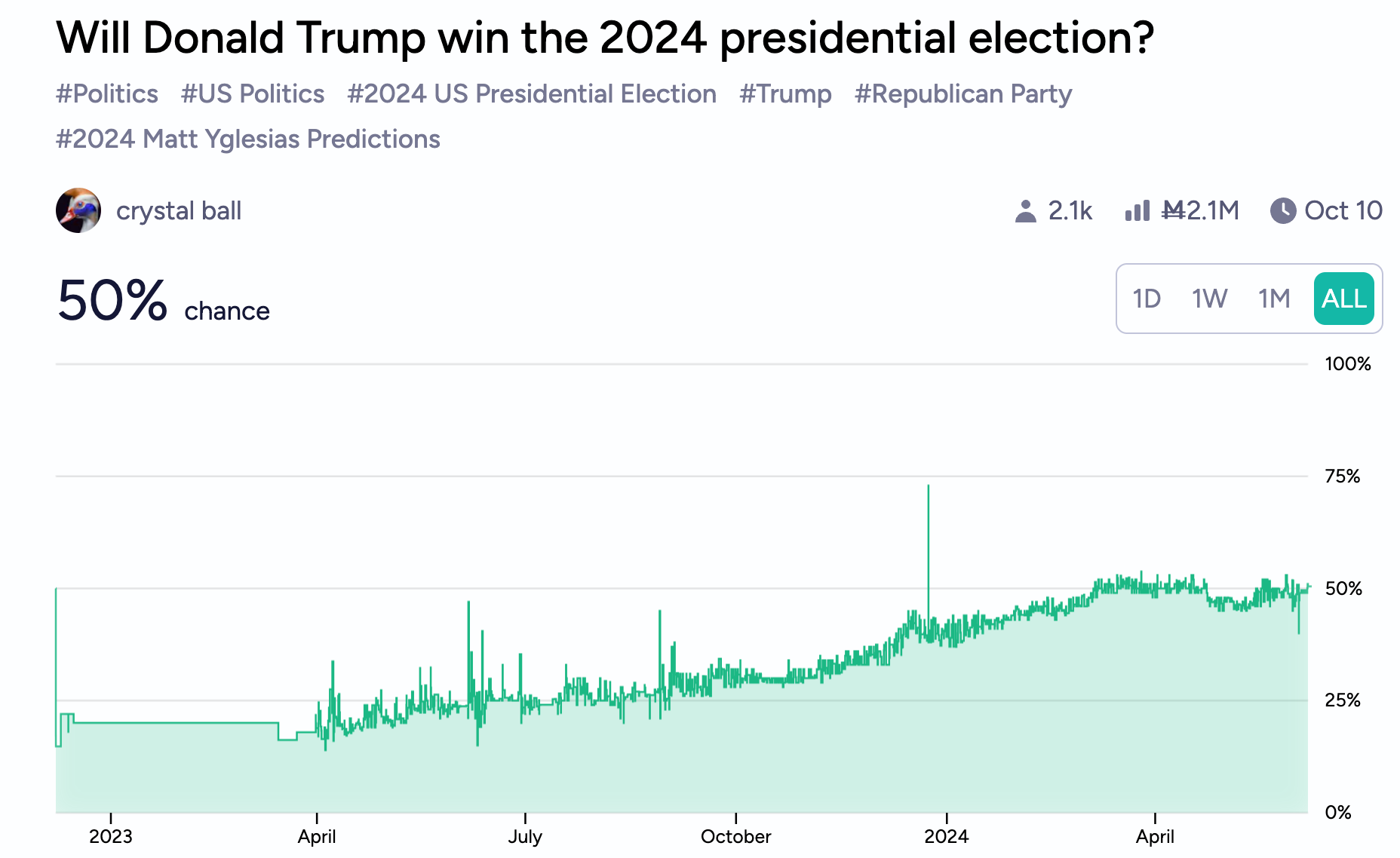}
    \label{fig:trump_mkt}
    \hspace{0cm}\begin{minipage}{0.83\textwidth}
    \vspace{0.3cm}
    \footnotesize \textit{Notes}. This figure shows an example of a market on the Manifold platform.
\end{minipage}
\end{figure}

\begin{figure}[H]
    \centering
    \hspace{0cm}\includegraphics[width=0.9\textwidth]{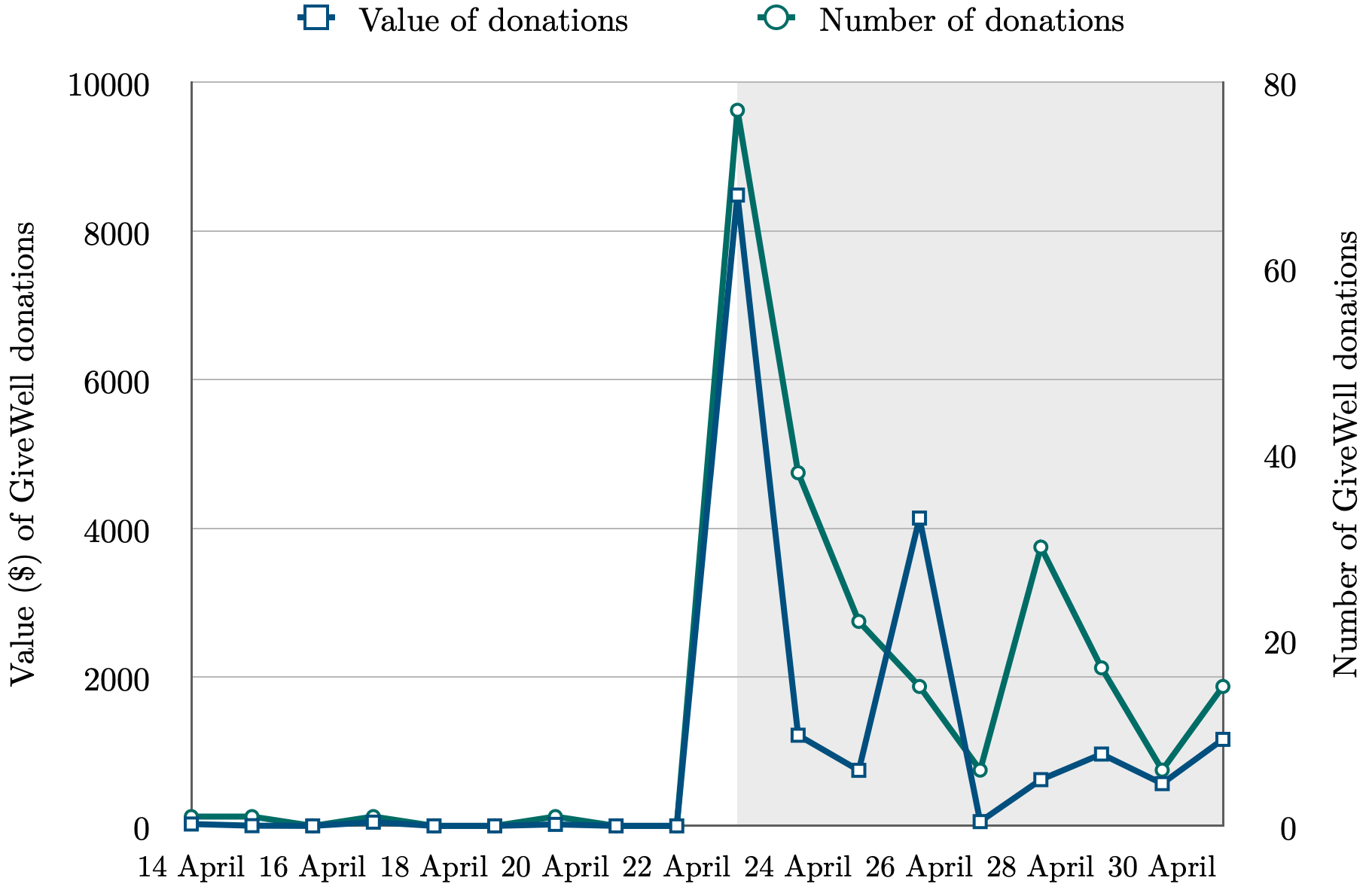}
    \caption{The impact of the Mana devaluation on GiveWell donations.}
    \label{fig:dev}
    \hspace{0cm}\begin{minipage}{0.8\textwidth}
    \vspace{0.3cm}
    \footnotesize \textit{Notes}. This figure shows the daily number and value (in USD) of donations to GiveWell's Maximum Impact Fund by Manifold users from 14 April 2024 to 1 May 2024. Dates after and including 23 April (the shaded area) represent dates after the announcement of the devaluation.
\end{minipage}
\end{figure}

%Less stretching
\renewcommand{\arraystretch}{1}

\begin{table}[H]
\begin{threeparttable}
\caption{Topics}
\vspace{-0.5em}
\label{tab:topics}
\begin{tabular}{lcc}\hline
Topic                                  & \hspace{1em}Frequency\hspace{1em} & \hspace{1em}Percentage\hspace{1em} \\\hline
Artificial intelligence                & 103   & 12.61   \\
US politics                            & 50    & 6.12    \\
Macroeconomics                         & 31    & 3.79    \\
Israel and the Palestinian territories & 30    & 3.67    \\
YouTube and its streamers              & 27    & 3.3     \\
Video games                            & 24    & 2.94    \\
Manifold Markets                       & 22    & 2.69    \\
The Russo-Ukrainian War                & 21    & 2.57    \\
Football                               & 19    & 2.33    \\
American football                      & 17    & 2.08    \\
Legal rulings                          & 16    & 1.96    \\
Television                             & 14    & 1.71    \\
Basketball                             & 12    & 1.47    \\
X (formerly Twitter)                   & 11    & 1.35    \\
Taylor Swift                           & 10    & 1.22    \\
UK politics                            & 10    & 1.22    \\
Cryptocurrency                         & 9     & 1.1     \\
Elon Musk                              & 9     & 1.1     \\
The 2024 Paris Olympics                & 9     & 1.1     \\
The Oscars                             & 9     & 1.1     \\
Russian politics                       & 9     & 1.1     \\
Apple Inc.                             & 8     & 0.98    \\
Science                                & 8     & 0.98    \\
Nuclear energy and weapons             & 7     & 0.86    \\
Satellites and space                   & 7     & 0.86    \\
Amazon.com, Inc.                       & 6     & 0.73    \\
Canadian politics                      & 6     & 0.73    \\
Effective altruism                     & 6     & 0.73    \\
Tesla, Inc.                            & 6     & 0.73    \\
Other topics                           & 203   & 36.87   \\
\hline
Total                                  & 719   & 100 \\\hline
\end{tabular}
\begin{tablenotes}
\footnotesize
\item \hspace{-0.2em}\textit{Notes}. This table shows the topics of the questions in the sample; the topic of each question is manually assigned and can be viewed using the replication package. Topics with 5 or fewer questions are omitted from the table.
\end{tablenotes}
\end{threeparttable}
\end{table}
%Back to normal
\renewcommand{\arraystretch}{1.22}

\begin{figure}[H]
    \centering
    \caption{Markets in the sample}\vspace{-0.5em}
    \hspace{0cm}\includegraphics[width=0.7\textwidth]{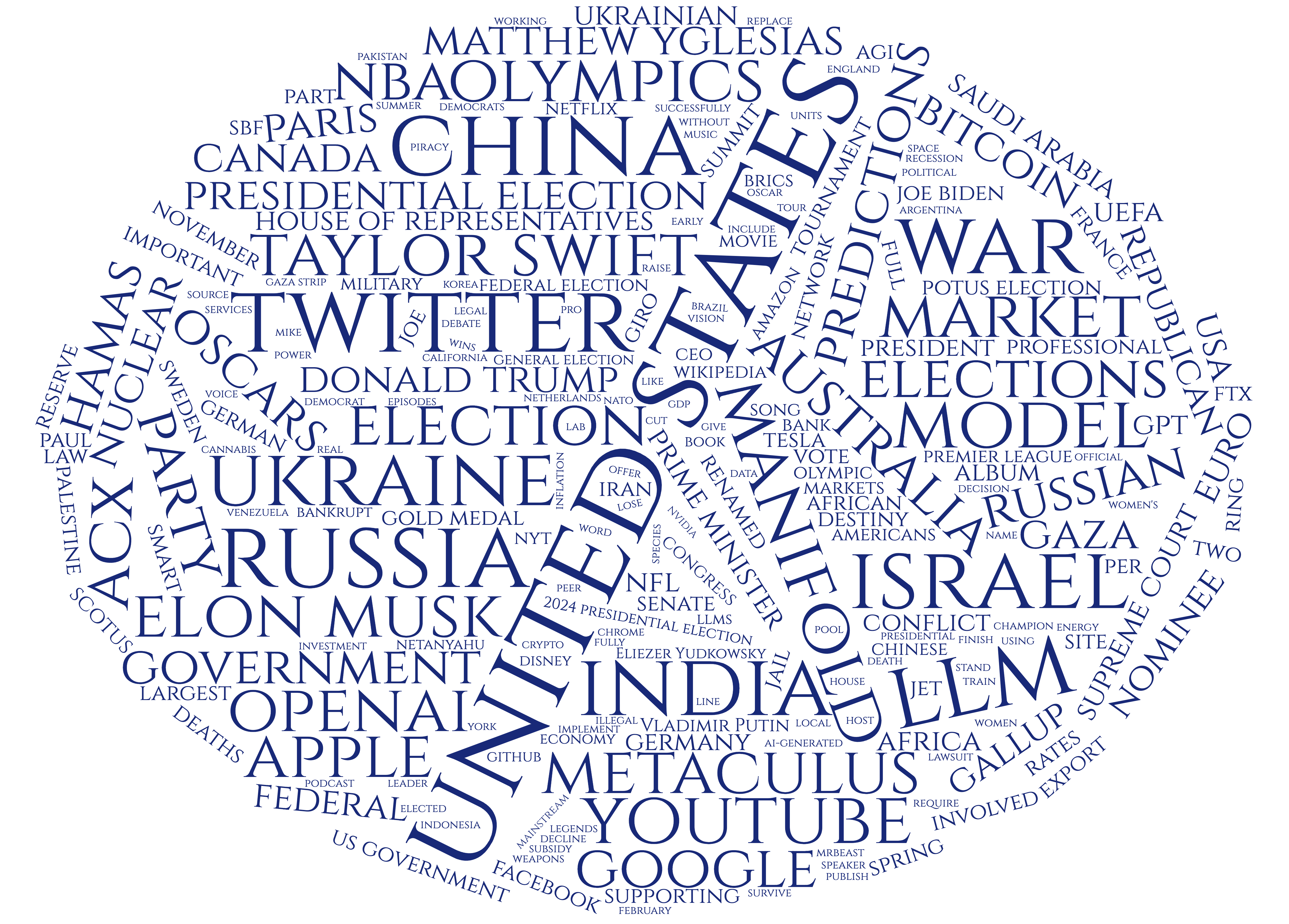}
    \label{fig:cloud}
    \hspace{0cm}\begin{minipage}{0.7\textwidth}
    \footnotesize \vspace{1em}\textit{Notes}. This figure shows a `word cloud' formed from the market questions in our sample.
\end{minipage}
\end{figure}

\begin{figure}[H]
    \centering
    \hspace{0cm}\includegraphics[width=0.9\textwidth]{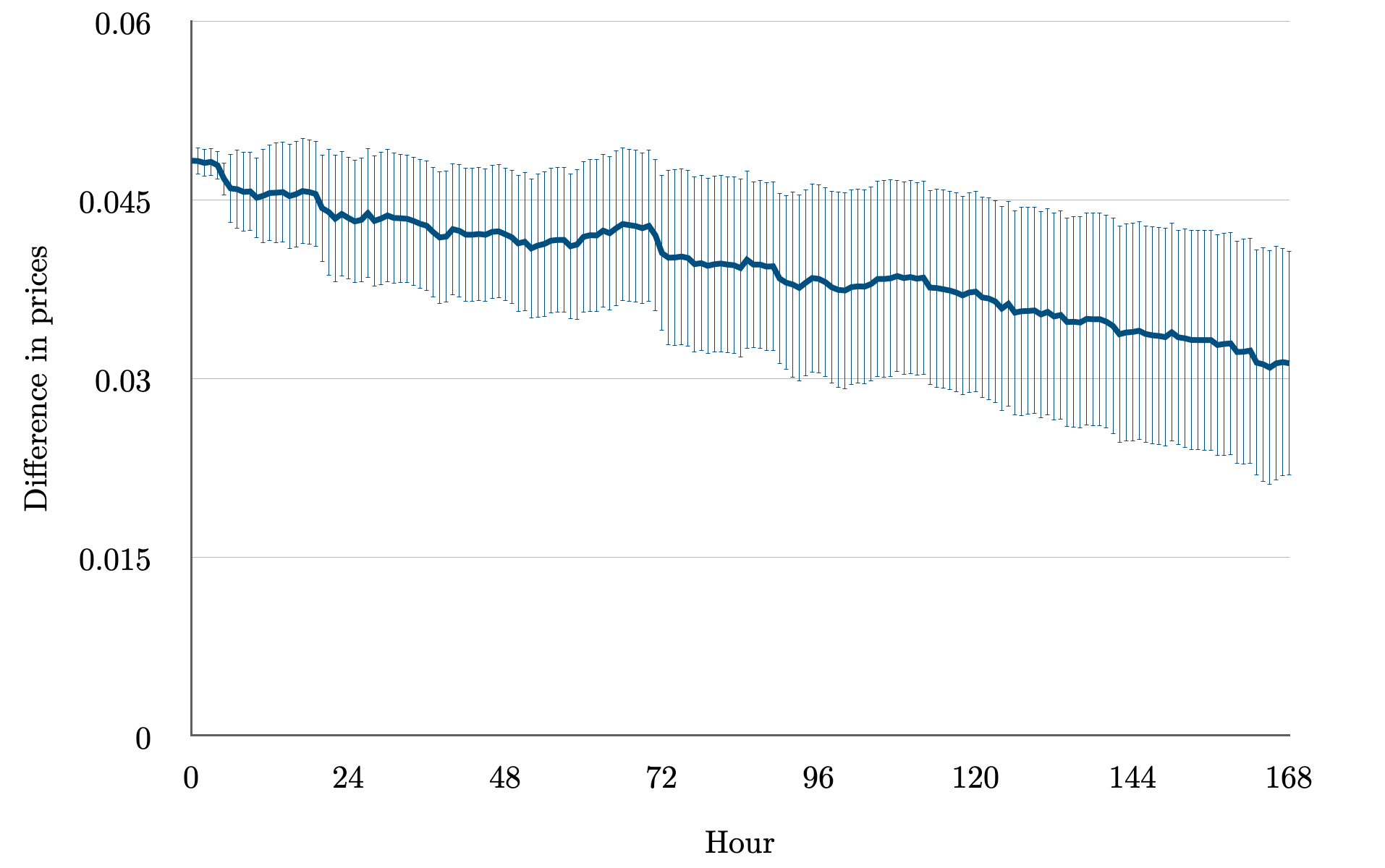}
    \caption{Comparing the `no' and `control' groups ($\hat{\beta}_2$)}
    \label{fig:no_control}
    \hspace{1.2cm}\begin{minipage}{0.72\textwidth}
    \vspace{0.2cm}
    \footnotesize \textit{Notes}. This figure plots the $\hat{\beta}_2$ coefficients obtained from estimating \eqref{eq_reg_main} for $t \in \{0, 1, ..., 167\}$ along with (robust) 95\% confidence intervals.
\end{minipage}
\end{figure}

\begin{figure}[H]
    \centering
    \hspace{0cm}\includegraphics[width=0.9\textwidth]{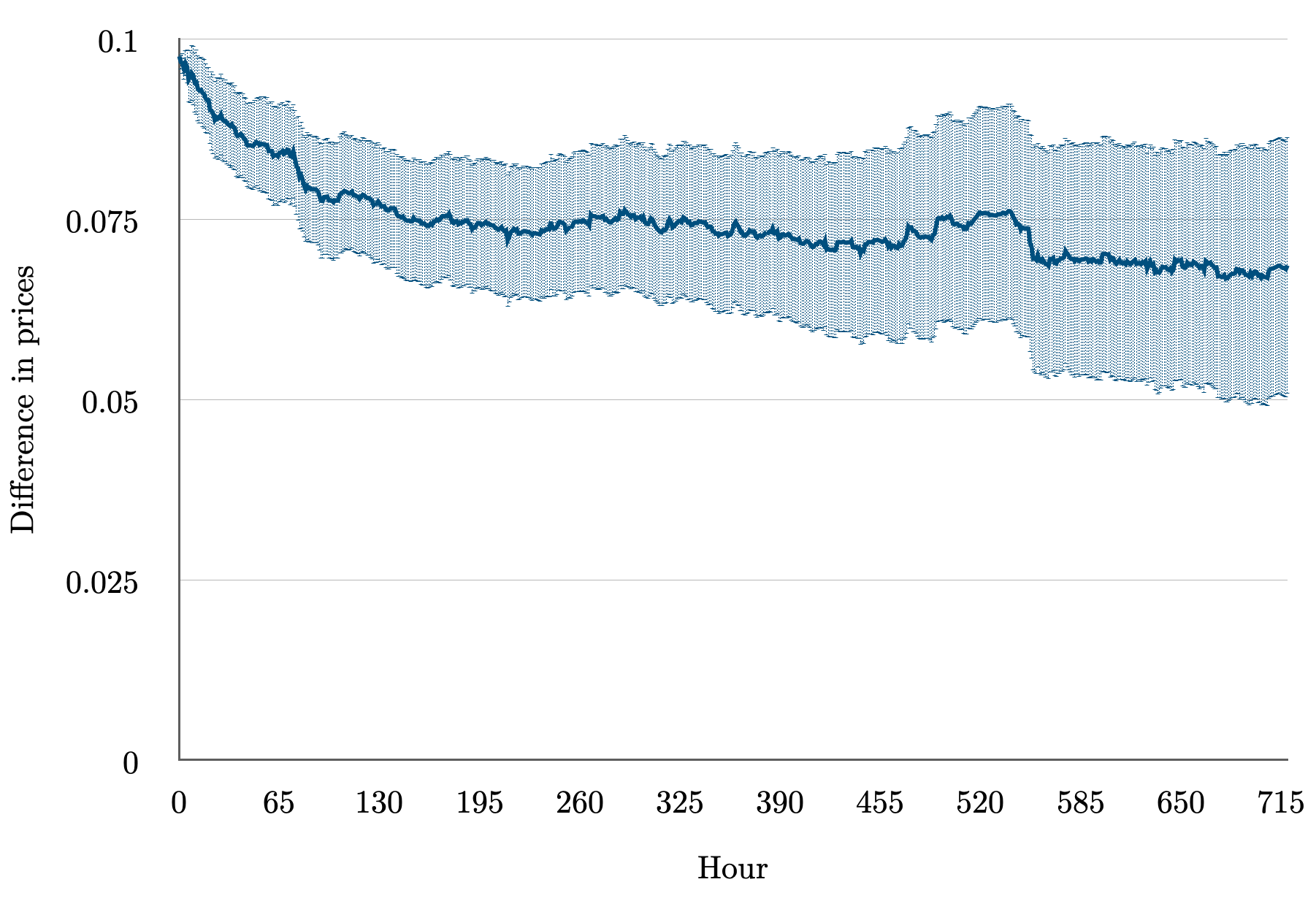}
    \caption{Comparing the `yes' and `no' groups ($\hat{\beta_1}$) after 30 days}
    \label{fig:yes_no_30d}
    \hspace{2cm}\begin{minipage}{0.8\textwidth}
    \vspace{0.3cm}
    \footnotesize \textit{Notes}. This figure plots the $\hat{\beta}_1$ coefficients obtained from estimating \eqref{eq_reg_main} for all $t \in \{0, 1, ..., 719\}$ along with (robust) 95\% confidence intervals.
\end{minipage}
\end{figure}

\begin{table}[H]
\begin{threeparttable}
\caption{Heterogeneity by number of traders}
\vspace{-0.5em}
\label{tab:trader_splits}
\begin{tabular}{lcccc}
\hline
            & (1)           & (2)           & (3)           & (4)           \\ \hline
Variable   & \hspace{0.9em}Full sample\hspace{0.9em}   & \hspace{0.9em}25\% split\hspace{0.9em}    & \hspace{0.9em}50\% split\hspace{0.9em}    & \hspace{0.9em}75\% split\hspace{0.9em}    \\
Yes         & 0.0748***     & 0.0693***     & 0.0666***     & 0.0625***     \\
            & {[}0.00446{]} & {[}0.00565{]} & {[}0.00602{]} & {[}0.00999{]} \\
Control     & 0.0313***     & 0.0268***     & 0.0264***     & 0.0259**      \\
            & {[}0.00484{]} & {[}0.00594{]} & {[}0.00673{]} & {[}0.0108{]}  \\
$p_{-1, i}$ & 1.000***      & 0.997***      & 1.009***      & 1.003***      \\
            & {[}0.00781{]} & {[}0.0103{]}  & {[}0.0157{]}  & {[}0.0283{]}  \\
Constant    & -0.0371***    & -0.0325***    & -0.0406***    & -0.0378**     \\
            & {[}0.00547{]} & {[}0.00719{]} & {[}0.00824{]} & {[}0.0149{]}  \\ \hline
$n$         & 817           & 640           & 419           & 207           \\
$R^2$       & 0.933         & 0.915         & 0.908         & 0.838         \\ \hline
\end{tabular}
\begin{tablenotes}
\footnotesize
\item \hspace{-0.2em}\textit{Notes}. This table shows the results of estimating regression \ref{eq_reg_main} on: the full sample (Column 1), the sample with an above 25th percentile number of traders (Column 2), the sample with an above median number of traders (Column 3), and the sample with an above 75th percentile number of traders (Column 4). Robust standard errors in parentheses (*** $p<0.01$, ** $p<0.05$, * $p<0.1$).
\end{tablenotes}
\end{threeparttable}
\end{table}

\begin{table}[H]
\begin{threeparttable}
\caption{Heterogeneity by total volume}
\label{tab:volume_splits}
\begin{tabular}{lcccc}
\hline
            & (1)           & (2)           & (3)           & (4)          \\ \hline
Variable    & \hspace{0.9em}Full sample\hspace{0.9em}   & \hspace{0.9em}25\% split\hspace{0.9em}    & \hspace{0.9em}50\% split\hspace{0.9em}    & \hspace{0.9em}75\% split\hspace{0.9em}   \\
Yes         & 0.0748***     & 0.0706***     & 0.0692***     & 0.0539***    \\
            & {[}0.00446{]} & {[}0.00485{]} & {[}0.00623{]} & {[}0.0106{]} \\
Control     & 0.0313***     & 0.0284***     & 0.0374***     & 0.0195*      \\
            & {[}0.00484{]} & {[}0.00575{]} & {[}0.00661{]} & {[}0.0112{]} \\
$p_{-1, i}$ & 1.000***      & 0.999***      & 1.012***      & 0.998***     \\
            & {[}0.00781{]} & {[}0.0103{]}  & {[}0.0124{]}  & {[}0.0245{]} \\
Constant    & -0.0371***    & -0.0359***    & -0.0479***    & -0.0315**    \\
            & {[}0.00547{]} & {[}0.00604{]} & {[}0.00699{]} & {[}0.0130{]} \\ \hline
$n$         & 817           & 613           & 410           & 205          \\
$R^2$       & 0.933         & 0.926         & 0.918         & 0.858        \\ \hline
\end{tabular}
\begin{tablenotes}
\footnotesize
\item \hspace{-0.2em}\textit{Notes}. This table shows the results of estimating regression \ref{eq_reg_main} on: the full sample (Column 1), the sample with an above 25th percentile total volume (Column 2), the sample with an above median total volume (Column 3), and the sample with an above 75th percentile total volume (Column 4). Robust standard errors in parentheses (*** $p<0.01$, ** $p<0.05$, * $p<0.1$).
\end{tablenotes}
\end{threeparttable}
\end{table}

\begin{table}[H]
\caption{Results from the Sweepcash markets}\label{tab:reg_sweepcash}
\begin{threeparttable}
\begin{tabular}{lccc}
\hline
Day \hspace{1em} & Full sample             & Full sample with controls & Reduced sample \\ \hline
0   & 0.0800***               & 0.0800***                                     & 0.0800***                          \\
    & {[}0.000{]}             & {[}0.000{]}                                   & {[}0.000{]}                        \\
1   & 0.0587***               & 0.0575***                                     & 0.0561***                          \\
    & {[}0.00934{]}           & {[}0.0104{]}                                  & {[}0.0111{]}                       \\
2   & 0.0537***               & 0.0537***                                     & 0.0481***                          \\
    & {[}0.00945{]}           & {[}0.0103{]}                                  & {[}0.0112{]}                       \\
3   & 0.0520***               & 0.0512***                                     & 0.0483***                          \\
    & {[}0.00960{]}           & {[}0.0103{]}                                  & {[}0.0112{]}                       \\
4   & 0.0463***               & 0.0447***                                     & 0.0401***                          \\
    & {[}0.0130{]}            & {[}0.0137{]}                                  & {[}0.0151{]}                       \\
5   & 0.0435***               & 0.0422***                                     & 0.0372**                           \\
    & {[}0.0138{]}            & {[}0.0147{]}                                  & {[}0.0160{]}                       \\
6   & 0.0463***               & 0.0455***                                     & 0.0392**                           \\
    & {[}0.0145{]}            & {[}0.0153{]}                                  & {[}0.0165{]}                       \\
7   & -                       & -                                             & 0.0348**                           \\
    & -                       & -                                             & {[}0.0169{]}                       \\ \hline
$n$ & 110 & 110                       & 89             \\ \hline
\end{tabular}
\begin{tablenotes}
\footnotesize
\item \hspace{-0.2em}\textit{Notes}. This table shows the results from the Sweepcash markets. The first column shows the $\hat{\beta_1}$ estimates obtained from regressing the price on day $t$ on the treatment dummy, controlling for $p_{-1, i}$. The second column shows the same results after adding controls for all other baseline variables. The third column shows results for the reduced sample of $89$ markets for which data is available over a 7 day horizon. Robust standard errors in parentheses (*** $p<0.01$, ** $p<0.05$, * $p<0.1$).
\end{tablenotes}
\end{threeparttable}
\end{table}

\clearpage

\section{Robustness}\label{sec_robustness}

In this section, we consider whether our results might be affected \textit{spillovers} between the markets. In theory, there are two ways in which this might occur:
\begin{itemize}
    \item \textit{Informational spillovers} arise when betting on one market provides information about another market, thereby influencing its price.
    \item \textit{Arbitrage spillovers} arise when betting on one market creates an arbitrage opportunity, thereby inducing trade on a closely related market.
\end{itemize}

Since we avoided betting on closely related markets, it is not clear that either type of spillover would have arisen during our experiment. For completeness, however, we now consider the possibility (and relevance) of such spillovers in some detail.

\textit{Theoretical observations.} In the main analysis, we estimate the model
\begin{equation}\label{eq_main_again}
p_{t, i} = \beta_0 + \beta_1 \mathbbm{1}_{Y, i} +  \beta_2 \mathbbm{1}_{C, i} + \beta_3 p_{-1, i} + u_i
\end{equation}
where $\mathbbm{1}_{Y, i}$ and $\mathbbm{1}_{C, i}$ are the treatment dummies, $p_{-1, i}$ is the price in market $i$ just before the bet, and $u_i$ captures the effect on subsequent trades on the price at time $t$. To make the analysis more tractable, we remove the control group from consideration and study the simpler model
\begin{equation}\label{eq_main_simple}
p_{t, i} = \beta_0 + \beta_1 \mathbbm{1}_{Y, i} + \beta_2 p_{-1, i} + u_i
\end{equation}
This model does not allow for the possibility of spillovers. To allow for this possibility, we suppose that a subset of the markets $S$ generates spillovers in another subset of the markets $S'$. For simplicity, we assume that $S$ and $S'$ are disjoint and let $\sigma_{i, j}$ denote the effect of market $j$ on market $i$. (Thus, $\sigma_{i, j} = 0$ if $i \notin S'$ or $j \notin S$.) We then generalise \eqref{eq_main_simple} to
\begin{equation}\label{eq_spillovers}
p_{t, i} = \beta_0 + \beta_1 \mathbbm{1}_{Y, i} + \beta_2 p_{-1, i} + \sum_{j \in S} \sigma_{i, j} \mathbbm{1}_{Y, j} - \sum_{j \in S} \sigma_{i, j} (1 - \mathbbm{1}_{Y, j}) +  u_i
\end{equation}
To understand \eqref{eq_spillovers}, suppose that market $i$ and market $j$ are positively correlated. In that case, betting \textit{yes} on $j$ should increase the price of $j$ and thereby increase the price of $i$: this explains the $\sigma_{i, j} \mathbbm{1}_{Y, j}$ term. For the same reason, however, betting \textit{no} on $j$ should decrease the price of $j$ and thereby decrease the price of $i$: this explains the $\sigma_{i, j} (1 - \mathbbm{1}_{Y, j})$ term. Given the random assignment of markets into either the yes group $Y$ (which increases $p_{t, i}$) or the no group $N$ (which decreases $p_{t, i}$), there is thus no reason to expect betting on $j$ to influence the price of $i$ on average. More formally, one sees that
\begin{equation}
\mathbb{E}[p_{t, i}] = \beta_0 + \beta_1 \mathbbm{1}_{Y, i} + \beta_2 p_{-1, i} + \sum_{j \in S} \sigma_{i, j} \mathbb{P}(j \in Y) - \sum_{j \in S} \sigma_{i, j} (1 - \mathbb{P}(j \in Y))
\end{equation}
and so the last two `spillover' terms cancel (given that $\mathbb{P}(j \in Y) = 0.5$, which in turn holds since the control group is omitted from the analysis).  In addition, it is worth noting that, while spillovers could be positive (in the case of positively correlated markets), they could just as well be negative --- which provides an additional reason for scepticism about whether they plausibly bias the results.

\textit{Empirical results.} The previous observations suggest that, even if spillovers did arise (which seems unlikely), they should not systematically bias the results. To investigate this issue empirically, we now consider the impact of dropping the `most related' markets from our sample. If spillovers are responsible for our results, then dropping these markets should substantially change our estimates. Conversely, if spillovers are absent (as we suspect) or irrelevant (as our theoretical analysis suggests), dropping these markets should not have a large impact on our estimates.

To do this, we considered sets of markets on which we had bet within a one-week period. Notice that, if one considers 7 day effects as we do here, it is not possible for spillovers to bias the results for markets whose betting times differ by more than 7 days. We gave every such set of questions to GPT-4.0 along with the prompt:
\begin{quote}
Some of these questions may be on the same topic. These questions have the property that, if you know the answer to one of them, you will learn a lot about the answer to another of them. If you find any such questions, can you flag them?
\end{quote}
In a handful of cases, GPT-4.0 flagged questions that were extremely unlikely to be connected via spillovers (e.g. `Will Jimmy Carter Die on a weekday in 2024?' \textit{vs} `Will Manifold require phone verification for all new users at any time in 2024?'). In such cases, we ignored GPT's suggestion. In all other cases, however, we dropped all but one of the flagged set of markets from the dataset, thereby generating a `reduced dataset' that can be viewed in the replication package.

Table \ref{tab:spillovers} displays the results of re-estimating our main analysis on the reduced dataset; to further reduce to the possibility of spillovers, we exclude the control group. As can be seen, dropping the (possibly) related markets reduces the sample from 556 to 482 markets. While this is not an especially large reduction, it should be remembered that we deliberately selected our markets to be unrelated. As can also be seen from comparing the two columns, the results are essentially identical after dropping the `most related' markets. This provides further evidence that spillovers are extremely unlikely to be generating our results.

\begin{table}[h]
\begin{threeparttable}
\caption{7 day effects on the reduced sample}
\vspace{-0.5em}
\label{tab:spillovers}
\begin{tabular}{lcc}
\hline
Variable  \hspace{2em}  & \hspace{2em}Main results\hspace{2em} & \hspace{1em}Reduced sample\hspace{1em} \\ \hline
Yes         & 0.075***        & 0.075***          \\
            & [0.004]        & [0.005]          \\
$p_{-1, i}$ & 1.003***        & 1.003***          \\
            & [0.008]        & [0.009]          \\
Constant    & -0.038***       & -0.037***         \\
            & [0.006]        & [0.006]          \\ \hline
$R^2$       & 0.937        & 0.942          \\
$n$           & 556          & 482            \\ \hline
\end{tabular}
\begin{tablenotes}
\footnotesize
\item \hspace{-0.2em}\textit{Notes}. This table shows the results of estimating regression \eqref{eq_main_simple} on the full sample (Column 1) and the reduced sample (Column 2), excluding the control group from all analyses. Robust standard errors in parentheses (*** $p<0.01$, ** $p<0.05$, * $p<0.1$).
\end{tablenotes}
\end{threeparttable}
\end{table}

\end{appendices}

\end{document}